\documentclass[twoside,11pt]{article}

\usepackage{blindtext}

\usepackage[preprint]{jmlr2e}
\usepackage{amsmath}
\usepackage{mathrsfs}

\usepackage{comment}
\usepackage{epsfig,morefloats}
\usepackage{amsfonts}
\usepackage{multirow}
\usepackage{enumitem}

\usepackage{amsthm}
\usepackage{color}
\usepackage{lscape}
\usepackage{rotating}
\usepackage{caption}
\usepackage{subfigure}
\usepackage{algorithm}
\usepackage{algpseudocode}
\usepackage{tikz}
\usepackage{subcaption}
\usepackage{float}
\usepackage{orcidlink}
\usetikzlibrary{arrows, positioning}

\newcommand{\BlackBox}{\rule{1.5ex}{1.5ex}}  
\ifdefined\proof
    \renewenvironment{proof}{\par\noindent{\bf Proof\ }}{\hfill\BlackBox\\[2mm]}
\else
    \newenvironment{proof}{\par\noindent{\bf Proof\ }}{\hfill\BlackBox\\[2mm]}
\fi
\theoremstyle{definition}
\newtheorem{assumption}{Assumption}
\newtheorem{example}{Example} 
\newtheorem{theorem}{Theorem}
\newtheorem{lemma}[theorem]{Lemma} 
\newtheorem{proposition}[theorem]{Proposition}

\newtheorem{definition}[theorem]{Definition}

\newcommand{\E}{\mathbb{E}}

\newcommand{\pa}{{\text{pa}}}
\newcommand{\nd}{{\text{nd}}}
\newcommand{\argmin}{{\rm argmin}}

\newcommand{\bH}{{\boldsymbol H}}
\newcommand{\bS}{{\boldsymbol S}}
\newcommand{\bU}{{\boldsymbol U}}
\newcommand{\bX}{{\boldsymbol X}}
\newcommand{\bZ}{{\boldsymbol Z}}
\usepackage{accents}

\usepackage{lastpage}
\jmlrheading{}{2025}{1-\pageref{LastPage}}{}{}{}{Li Chen, Xiaotong Shen and Wei Pan}

\ShortHeadings{Enhancing Causal Effect Estimation with Diffusion-Generated Data}{Chen, Shen and Pan}
\firstpageno{1}

\begin{document}

\title{Enhancing Causal Effect Estimation with Diffusion-Generated Data\thanks{Corresponding author: Xiaotong Shen.}}

\author{\name Li Chen \email chen7019@umn.edu \\
       \addr School of Statistics\\
       University of Minnesota\\
       Minnesota, MN, 55455
       \AND
       \name Xiaotong Shen \textsuperscript{\orcidlink{0000-0003-1300-1451}} \email xshen@umn.edu \\
       \addr School of Statistics\\
       University of Minnesota\\
       Minnesota, MN, 55455
       \AND
       \name Wei Pan \email panxx014@umn.edu \\
       \addr Division of Biostatistics\\
       University of Minnesota\\
       Minnesota, MN, 55455
       }


\maketitle

\begin{abstract}
Estimating causal effects from observational data is inherently challenging due to the lack of observable counterfactual outcomes and even the presence of unmeasured confounding. Traditional methods often rely on restrictive, untestable assumptions or necessitate valid instrumental variables, significantly limiting their applicability and robustness. In this paper, we introduce Augmented Causal Effect Estimation (ACEE), an innovative approach that utilizes synthetic data generated by a diffusion model to enhance causal effect estimation. By fine-tuning pre-trained generative models, ACEE simulates counterfactual scenarios that are otherwise unobservable, facilitating accurate estimation of individual and average treatment effects even under unmeasured confounding. Unlike conventional methods, ACEE relaxes the stringent unconfoundedness assumption, relying instead on an empirically checkable condition. Additionally, a bias-correction mechanism is introduced to mitigate synthetic data inaccuracies. We provide theoretical guarantees demonstrating the consistency and efficiency of the ACEE estimator, alongside comprehensive empirical validation through simulation studies and benchmark datasets. Results confirm that ACEE significantly improves causal estimation accuracy, particularly in complex settings characterized by nonlinear relationships and heteroscedastic noise.
\end{abstract}

\begin{keywords}
  causal effect estimation, unmeasured confounding, directed acyclic graphs, generative models, knowledge transfer
\end{keywords}

\section{Introduction}

Estimating the effect of interventions or treatments on outcomes is fundamental to causal inference. While randomized controlled trials are considered the gold standard, ethical or practical constraints often limit their feasibility. Consequently, observational studies become essential for causal analysis. Modern scientific research frequently involves large, diverse datasets, underscoring the need for advanced statistical methods to integrate and transfer knowledge from auxiliary datasets to enhance causal investigations. 

Observational data typically contain confounders—variables influencing both treatment and outcome—that can bias causal estimates. When these confounders are measured, methods such as matching and weighting are widely used \citep{rosenbaum1983central, horvitz1952generalization}. Matching methods, including nearest neighbor matching, pair treated units with similar control units based on covariates, and estimate causal effects by comparing outcomes within matched pairs \citep{abadie2011bias, heckman1998matching, lin2023estimation}. However, these methods struggle when treated and control groups have limited overlap and depend heavily on matching metrics and procedures. Weighting methods attempt to balance covariate distributions by creating a pseudo-population but can be sensitive to model misspecifications and extreme weights \citep{horvitz1952generalization, robins1994estimation, chan2016globally}. Therefore, robust methods that handle limited covariate overlap without restrictive assumptions are crucial.

Addressing unmeasured confounding typically involves instrumental variable (IV) analyses or sensitivity analyses \citep{angrist1996identification, imbens2003sensitivity}. IV methods rely on variables affecting treatment but not directly influencing outcomes \citep{angrist1996identification, kang2016instrumental, chen2024discovery}. However, identifying valid IVs requires extensive domain knowledge, and critical assumptions—such as instrument exogeneity (the IV is independent of the error term) and exclusion restriction (the IV affects the outcome only through the treatment)—are challenging to verify empirically. Violations of these assumptions can severely bias causal estimates \citep{guo2023causal}. In numerous economic and biological applications, confounders often influence many observed variables, a phenomenon known as pervasive confounding that can be utilized to deconfound. Unfortunately, most existing methods that tackle causal relationships with pervasive confounding assumption typically assume linear causal effects\citep{frot2019robust,shah2020right,guo2022doubly}, with an exception in the Causal Additive Model\citep{agrawal2023decamfounder}. A notable nonlinear causal effect estimation method incorporating deconfounding adjustment is proposed in \cite{li2024nonlinear}; however, this approach depends on a sublinear growth condition for nonlinear functions with additive confounders and noise, a restriction that remains limiting.

This paper introduces a novel approach for causal effect estimation that leverages advanced generative modeling techniques while relaxing conventional assumptions and model restrictions. Our method generates synthetic samples using a diffusion model designed to closely mimic the original data. The fidelity of these samples is further enhanced by fine-tuning a pre-trained model on auxiliary datasets from related studies, thereby facilitating effective knowledge transfer. By simulating unobserved counterfactual scenarios—a common limitation of traditional observational studies where each individual receives only one treatment—this approach enables direct comparisons across different treatment conditions. Furthermore, by bridging the gap between conditional generation and causal effect estimation under mild assumptions and incorporating a bias correction mechanism to reconcile discrepancies between synthetic and original data, our framework robustly addresses unmeasured confounding and enhances the statistical validity and reliability of the resulting causal estimates.

 The contributions of this paper are as follows:
\begin{itemize}
    \item Methodology: We introduce a novel method called Augmented Causal Effect Estimation (ACEE), which leverages synthetic data generated by a conditional diffusion model to enhance causal effect estimation. Toward this goal, we bridge the gap between causal effect estimation and conditional generation with the presence of unmeasured confounding under mild conditions, making it feasible to transfer knowledge from a generative model pre-trained on auxiliary data. Moreover, ACEE incorporates a bias-correction mechanism designed to ensure statistical validity even when significant distributional discrepancies exist between synthetic samples and the original data.

    \item Statistical guarantee: We establish theoretical consistency for the ACEE estimator, demonstrating that it converges to the true causal effects as the conditional generative model approaches the true conditional distribution. Additionally, the bias-corrected ACEE estimator maintains consistency even if the synthetic data are of lower fidelity. Notably, ACEE relaxes many stringent assumptions required by traditional methods when the conditional generative model performs adequately.
    
    \item Empirical validation: Comprehensive numerical studies illustrate the practical advantages and superior performance of ACEE, showing that it compares favorably against several leading competitors from existing literature.
\end{itemize}

The remainder of this paper is organized as follows. Section 2 introduces the causal effect estimation method and develops a bias-corrected estimator. Section 3 extends the approach to estimate total causal effects within directed acyclic graphs (DAGs) that accommodate unmeasured confounding. Section 4 examines the theoretical properties of the ACEE estimator. Section 5 presents numerical studies, including simulations, analyses of a benchmark dataset, and experiments on a pseudo-real dataset. Finally, Section 6 concludes the article. Additional discussions and technical proofs are provided in the Appendix.

\section{Causal Effects Estimation}

Causal inference is frequently formulated within the potential outcomes framework introduced in \cite{splawa1990application,rubin1974estimating,holland1986statistics}. In this framework, each individual possesses two potential outcomes, namely \(Y(1)\) if treated and \(Y(0)\) if untreated, although only one of them is observed. Let \(D \in \{0,1\}\) denote the binary treatment indicator, with \(D=1\) corresponding to treatment and \(D=0\) corresponding to control. The observed outcome is then given by
$Y = D\,Y(1) + (1-D)\,Y(0)$.
Given independent and identically distributed samples \(\{(\bX_i, D_i, Y_i)\}_{i=1}^n\), each individual $i$ has a vector of observed covariates $\bX_i$, the treatment assignment $D_i$ and the outcome $Y_i = D_i\,Y_i(1) + (1-D_i)\,Y_i(0)$. 

\subsection{Individual and Average Treatment Effects} 

Our objective is to estimate the individual treatment effect (ITE) and the average treatment effect (ATE), defined as the difference between the potential outcomes under treatment and control conditions for any individual $i$ with specific pretreatment characteristics, and to aggregate these individual differences to quantify the overall impact of the treatment across the population. The estimation of ITE and ATE is often conducted under the unconfoundedness assumption \citep{rubin1974estimating,rosenbaum1983central}, given the inherent challenge that counterfactual outcomes are unobserved. For example, \cite{cai2024conformal} presents a diffusion model-based conformal inference approach for ITE estimation under this
assumption.  

This unconfoundedness assumption requires that treatment assignment behaves as if it were random when conditioned on covariates $\bX$, ensuring conditional independence of $Y$ and $D$ given $\bX$. Conditioning on $\bX$ eliminates confounding bias, allowing us to estimate the ITE as  
$\tau(\bX) = \E[Y(1) \mid \bX] - \E[Y(0) \mid \bX]$
and the average treatment effect (ATE) as $\tau = \E[\tau(\bX)]$.

However, the unconfoundedness assumption is inherently untestable and often unrealistic in practice \citep{heckman1990varieties,imbens2003sensitivity,hernan2020causal}. Critically, it does not apply to scenarios involving unmeasured confounders. To address these limitations, we introduce an assumption termed conditional randomness, which explicitly models the confounding effects induced by unobserved confounders. This assumption is well-suited to our proposed approach, which leverages generative models to create paired synthetic data replicating the distribution of the original data. Specifically, conditional randomness states that treatment assignment becomes random when conditioned on the observed covariates and a latent confounder vector $\bH$. The vector $\bH$ is exogenous concerning $\bX, D$ and $Y$, meaning it can only influence, rather than be influenced by, covariates, treatment assignment and potential outcomes.

\begin{assumption}[Conditional Randomness]
\label{assu:conditional_random}
$Y$ and $D$ are conditionally independent given $(\bX,\bH)$.
\end{assumption}

Assumption \ref{assu:conditional_random} is widely adopted in the literature of causal inference with unmeasured confounders \citep{jin2023sensitivity,yadlowsky2022bounds} and as pointed out by \cite{yadlowsky2022bounds}, such a latent confounder vector $\bH$ should almost always exist. When unconfoundedness holds, conditioning on $\bX$ alone is sufficient, and the latent confounder vector 
$\bH$ can be set to be an empty set, meaning conditional randomness reduces to unconfoundedness. However, in settings with unmeasured confounders, incorporating 
$\bH$ allows us to account for residual bias not captured by $\bX$ alone. This generalization enhances causal inference by providing a more flexible framework that accommodates the complexities of real-world data.

Ideally, one may define the ITE as $\E[Y(1) \mid \bX, \bH] - \E[Y(0) \mid \bX, \bH]$. However,
since the latent confounder \(\bH\) is unobserved and generally non-identifiable without further assumptions, we propose to approximate it using a proxy vector $\bS$, defined as $\bS = E[\bZ \mid \bH]$, with observed $\bZ = (\bX^\top,D, Y)^{\top}$. The proxy vector $\bS$ represents the projection of $\bZ$ onto the space spanned by $\bH$, which can be estimated by \(\widehat \bS\) using various methods such as principal component analysis \citep{fan2013large}, diversified projection \citep{fan2023factor}, or deep autoencoder models \citep{zhu2021deeplink}. See Section \ref{Sestimation} for an illustration of estimation.

By using $\bS$ to substitute for \(\bH\), we can now define ITE as  
$\tau(\bX,\bS) = E[Y(1)\mid \bX,\bS] - E[Y(0)\mid \bX,\bS]$.
This allows us to adjust for hidden confounding and estimate $\tau$ by conditioning on $(\bX, \bS)$ rather than on the unobserved $\bH$. A natural question arises that whether this definition is equivalent to $\E[Y(1) \mid \bX, \bH] - \E[Y(0) \mid \bX, \bH]$ and how to estimate this defined quantity. To ensure that the proxy $\bS$ fully captures the influence of the latent confounders, we impose the following assumption, which helps to relax the unconfoundedness assumption.

\begin{assumption}[Proxy Sufficiency]
\label{assu:proxy} 
Denote $S_Y = \E[Y \mid \bH]$ and $\bS = (\bS_{-Y}^\top,S_Y)^\top$, the residual \( Y - S_Y \) is conditionally independent of \( \bH \) given \( (\bX, \bS_{-Y}, D) \).
\end{assumption}

Unlike the unconfoundedness condition, the proxy sufficiency assumption 
can be empirically evaluated to some extent using a diagnostic tool that examines whether the fitted residual \(Y - S_Y\) behaves randomly with respect to the estimated $S_Y$---i.e., without systematic patterns---when conditioned on \((\bX, \bS_{-Y}, D)\) as $S_Y$ is a function of $\bH$. More importantly, this assumption enables the estimation of both ITE and ATE even in the presence of unmeasured confounders, a scenario where methods solely relying on unconfoundedness would fail. An example illustrating a scenario in which Assumption \ref{assu:proxy} holds is provided below.

\begin{example}
	\label{example:proxy}
	Consider the following causal model with additive unmeasured confounders:
	\begin{equation*}
		\begin{split}
			Y & = f(X_1,\ldots,X_p,D) + g(\boldsymbol{H}) + \varepsilon, \\
			X_1 & = g_1(\boldsymbol{H}) + \varepsilon_1, \\
			X_j & = f_j(X_1,\ldots, X_{j-1}) + g_j(\boldsymbol{H}) + \varepsilon_j, \quad \text{for } j = 2, \ldots, p, \\
			P(D=1 \mid \boldsymbol{X},\boldsymbol{H}) & = \frac{1}{\Bigl(1+\exp\bigl(f_{p+1}(X_1,\ldots,X_p)\bigr)\Bigr)
				\Bigl(1+\exp\bigl(g_{p+1}(\boldsymbol{H})\bigr)\Bigr)},
		\end{split}
	\end{equation*}
	where \(\varepsilon, \varepsilon_1, \ldots, \varepsilon_p\) are independent noise terms with mean zero. As shown in the Appendix, one can verify that $Y - \E[Y \mid \bH]$ is conditionally independent of $\bH$ given $(\bX,\E[\bX \mid \bH], \E[D \mid \bH],D)$, so that Assumption \ref{assu:proxy} holds for this model.
\end{example}

\begin{lemma}
\label{lem:proxy}
Under Assumption \ref{assu:proxy}, the individual treatment effect (ITE) can be expressed as  
\[
\tau(\bX,\bS) = \E[Y \mid \bX, \bS, D = 1] - \E[Y \mid \bX, \bS, D = 0]
=\E[Y(1) \mid \bX, \bH] - \E[Y(0) \mid \bX, \bH].
\]  
\end{lemma}

Aggregating over the population, we obtain ATE
 $\tau = E[\tau(\bX,\bS)]$. Lemma \ref{lem:proxy} suggests that by substituting the unobserved confounder \(\bH\) with its proxy $\bS$, we retain the essential information necessary for the identification and estimation of the ITE, thus ATE. 

 Define the response functions
 $\mu_d(\bX, \bS) = \mathbb{E}[Y(d) \mid \bX, \bS]$;
 $d=0,1$, then it follows that $\tau(\bX,\bS) =\mu_1(\bX, \bS)-\mu_0(\bX, \bS)$.

Empirically, synthetic samples are generated by a conditional diffusion model:
$Y^{(m)}(d) \sim P(Y\mid \bX, \widehat \bS, D=d)$; $d=0,1$, 
and the corresponding estimates are given by
\begin{align}
	\widehat{\mu}_d(\bX, \widehat{\bS}) &= \frac{1}{M} \sum_{m=1}^M Y_i^{(m)}(d), \quad \text{for } d=0,1.
\end{align}
Then the ITE can be estimated as 
\[
\widehat \tau(\bX,\widehat{\bS}) = \widehat{\mu}_1(\bX, \widehat{\bS})-\widehat{\mu}_0(\bX, \widehat{\bS}),
\]
We denote $\widehat{\bS}_i$ as the corresponding proxy to individual $i$ with observed variables $(\bX_i,D_i,Y_i)$, then the estimated ATE becomes 
\begin{equation}\label{eq:ATE_empirical}
	\widehat{\tau} = \frac{1}{n} \sum_{i=1}^n \widehat \tau(\bX_i,\widehat \bS_i) 
\end{equation}

The deconfounder method proposed in \cite{wang2019blessings} may seem to be similar to our approach at the first glance by bypassing the unconfoundedness assumption with a substitute confounder. However, their key assumption ``no single-cause confounders" reduces to the ``no unmeasured confounders" assumption in our single-cause setting.

\subsection{Estimation of the Proxy $\bS$}
\label{Sestimation}

This section details the estimation procedure for the proxy $\bS$. Assume that the collection \(\{E[Z_l \mid \bH]\}_{l=1}^{p+2}\) has rank \(q\) with \(q \le p+2\). Then there exist \(q\) functions \(\{\varphi_\ell\}_{\ell=1}^q\) (with \(1 \le q < \infty\)) such that for each \(l = 1, \dots, p+2\),
$E[Z_l \mid \bH] = \psi_l^\top \Phi(\bH)$,
where  
$\Phi(\bH) = (\varphi_1(\bH), \cdots, \varphi_q(\bH))^{\top} \in \mathbb{R}^q$.
In matrix notation, this relationship is written as
\[
\boldsymbol{Z} = \Psi\,\Phi(\bH) + \bU,
\]
with \(\Psi \in \mathbb{R}^{(p+2)\times q}\) and an error term \(\bU\) that is uncorrelated with \(\Phi(\bH)\). For identification, we impose the constraints:
$\operatorname{Cov}\bigl(\Phi(\bH)\bigr)= I_q \quad \text{and} \quad \Psi^\top \Psi \text{ is diagonal.}$

Let 
$\mathbf{Z} =(\boldsymbol{Z}_1, \cdots, \boldsymbol{Z}_n)^{\top}\in \mathbb{R}^{n\times (p+2)}$
denote the observed data matrix, and let \(\mathbf{H}\) represent the corresponding unobserved confounding variables. The constrained least squares estimator \((\hat{\Phi}(\mathbf{H}), \hat{\Psi})\) is then defined as
\[
(\hat{\Phi}(\mathbf{H}), \hat{\Psi}) = \arg\min_{\Phi(\mathbf{H})\in\mathbb{R}^{n\times q},\,\Psi\in\mathbb{R}^{(p+2)\times q}} \| \mathbf{Z} - \Phi(\mathbf{H})\Psi^\top \|_F^2,
\]
subject to
$\frac{1}{n}\Phi(\mathbf{H})^\top \Phi(\mathbf{H}) = I_q$ and  $\Psi^\top \Psi \text{ is diagonal}$,
where $\|\cdot\|_F$ denotes the Frobenius norm. Finally, the proxy for the latent confounder is estimated by
\[
\hat{\mathbf{S}} = \hat{\Phi}(\mathbf{H})\hat{\Psi}^\top.
\]

\subsection{Bias Correction}
\label{bias-correction}

In some scenarios, the conditional generative model used to generate potential outcomes may have low-fidelity to the original data, leading to biased estimates of the individual treatment effect (ITE) and the average treatment effect (ATE). In this subsection, we introduce a bias-correction method that adjusts the counterfactual outcomes by leveraging information from the nearest neighbors.

Let $\mathcal{J}_N^0(\bX,\bS)$ and $\mathcal{J}_N^1(\bX,\bS)$ denote the index set of the \(N\) nearest neighbors in $(\bX_i,\widehat{\bS}_i)_{i=1}^n$ of the pair $(\bX,\bS)$ among those control and treated observations, respectively. Then $\mathcal{J}_N^0(\bX,\bS) \subset \{ i : D_i = 0\}$ and $\mathcal{J}_N^1(\bX,\bS) \subset \{ i : D_i = 1\}$. The idea is to use the residuals from these neighbors to correct the bias in the estimated response functions.

For the control response function, the bias-corrected estimator is defined as
\[
\widehat{\mu}^{c}_{0}(\bX,\bS) = \widehat{\mu}_{0}(\bX,\bS) + \frac{1}{N} \sum_{i \in \mathcal{J}_N^0(\bX,\bS)} \Bigl(Y_i - \widehat{\mu}_{0}(\bX_i, \widehat{\bS}_i)\Bigr)
\]
That is, the initial estimated control response function \(\widehat{\mu}_{0}(\bX,\bS)\) is adjusted by the average residual (the difference between observed and predicted outcomes) computed from its \(N\) nearest neighbors in the control group.

Similarly, for the treated response function, we define the bias-corrected estimator as
\[
\widehat{\mu}^{c}_{1}(\bX,\bS) = \widehat{\mu}_{1}(\bX,\bS) + \frac{1}{N} \sum_{i \in \mathcal{J}_N^1(\bX,\bS)} \Bigl(Y_i - \widehat{\mu}_{1}(\bX_i, \widehat{\bS}_i)\Bigr)
\]
Here, the initial estimated treated response function \(\widehat{\mu}_{1}(\bX,\bS)\) is adjusted by the average residual (the difference between observed and predicted outcomes) computed from its \(N\) nearest neighbors in the treated group.

The bias-corrected individual treatment effect (ITE) of individual $i$ with the pair $(\bX_i,\bS_i)$ is then given by the difference between the corrected treated and control responses:
\[
\widehat{\tau}^c(\bX_i,\widehat{\bS}_i) = \widehat{\mu}^{c}_{1}(\bX_i,\widehat{\bS}_i) - \widehat{\mu}^{c}_{0}(\bX_i,\widehat{\bS}_i).
\]

We define the matching count: $K_N(i) = \sum_{j=1}^n \mathbf{1}\{ i \in \mathcal{J}_N^{D_i}(\bX_j,\widehat{\bS}_j) \}$.
This count reflects how many times observation \(i\) serves as one of the \(N\) nearest neighbors for units in the observational data points. Then the overall bias‐corrected ATE estimator can be expressed as:
\begin{equation}\label{eq:estimator_bc}
\begin{split}
\widehat{\tau}^{c} &= \frac{1}{n} \sum_{i=1}^{n} \widehat{\tau}(\bX_i, \widehat{\bS}_i) 
= \widehat{\tau} + \frac{1}{n} \Biggl[\sum_{i: D_i=1} \frac{K_N(i)}{N}\widehat{R}_i - \sum_{i: D_i=0} \frac{K_N(i)}{N}\widehat{R}_i\Biggr],
\end{split}
\end{equation}
where the residual for observation \(i\) is defined as $\widehat{R}_i = Y_i - \widehat{\mu}_{D_i}(\bX_i, \widehat{\bS}_i)$.

In summary, our bias-correction framework refines synthetic outcomes by adjusting for discrepancies observed among similar units, thereby improving the accuracy of both ITE and ATE estimates. Prior approaches, such as those studied in \cite{abadie2011bias,lin2023estimation}, similarly employ bias correction but strictly rely on the unconfoundedness assumption, utilizing observed outcomes for the imputation in their assigned treatment groups directly without correction. In contrast, our methodology conducts bias correction under Assumption \ref{assu:proxy}, which substantially relaxes the stringent unconfoundedness condition. Our approach applies bias correction uniformly to neighborhood observations, regardless of whether outcomes are observed or missing, thus facilitating the establishment of consistency for ITE estimation.

\section{Casual Effects Estimation under Directed Acyclic Graph}

\subsection{Causal Effect Estimation over a DAG with Unmeasured Confounding}

Consider a \(p\)-dimensional random vector \(\bX = (X_1, \dots, X_p)\). In the presence of unobserved confounders, we assume that the complete set of variables \((\bX, \bH)\) adheres to the following structural causal model:
\begin{equation}
    \label{eq:sem}
    X_j = f_j\bigl(\bX_{\mathrm{pa}(j)}, \bH, \epsilon_j\bigr), \quad j \in \mathcal{V} = \{1,\ldots,p\},
\end{equation}
where \(\mathrm{pa}(j)\) is a parent index set consisting of the parent variables of \(X_j\) excluding $X_j$ itself, as determined by the parent-child relation explicitly defined by the function \(f_j\), the parent vector \(\bX_{\mathrm{pa}(j)}\), the unmeasured vector \(\bH\), and the noise \(\epsilon_j\). The noise \(\epsilon_j\) is assumed to be independent of \((\bX_{\mathrm{pa}(j)}, \bH)\). A justification of assuming such a causal structural model is provided in the Appendix \ref{sec:canonical_DAG}.

To ensure that the effect of each parent is non-vanishing, we assume causal minimality in \eqref{eq:sem}. In particular, for every \(j \in \mathcal{V}\) and for each \(k \in \mathrm{pa}(j)\), the function \(f_j\) is assumed to depend non-trivially on its \(k\text{th}\) argument. That is, there exists some choice of the remaining inputs \(\bX_{\mathrm{pa}(j)\setminus\{k\}}, \bH,\) and \(\epsilon_j\) such that the mapping
$x_k \mapsto f_j\bigl(x_k,\, \bX_{\mathrm{pa}(j)\setminus\{k\}}, \bH, \epsilon_j\bigr)$
is not constant. In other words,
if there exists \(\mathcal{B} \subseteq \pa(j)\) that the value of $f_j$ only depends on \((\bX_{\mathcal{B}},\bH,\epsilon_j)\), then $\mathcal{B} = \pa(j)$. Thus \eqref{eq:sem} encodes a directed graph \(\mathcal{G} = (\mathcal{V}, \mathcal{E})\), such that $\mathcal{E} = \{k \rightarrow j: k \in \pa(j), j \in \mathcal{V}\}$. Furthermore, we assume that \(\mathcal{G}\) is a directed acyclic graph (DAG) that no directed path $k \rightarrow \cdots \rightarrow k$ exists in \(\mathcal{G}\).  

 For any DAG, there at least exists a topological ordering \(\pi\) of the indices \(\{1, \dots, p\}\) such that for every directed edge \(k \to j\), the ordering satisfies \(\pi(k) < \pi(j)\). Now consider the causal effect of \(X_k\) on $X_j$ when \(X_k\) precedes $X_j$ in the causal ordering \(\pi\). Define \(\bX_{k^-}\) as the set of variables preceding $X_k$ in \(\pi\) and \(\bX_{k^+} = \bX_{k^-} \cup \{X_k\}\). The total causal effect of \(X_k\) on $X_j$ is formally defined as:
\begin{equation}
\label{eq:DAG_ate_conditionalH}
\tau(k, j; x_1, x_0) = \mathbb{E}\left[\mathbb{E}[X_j \mid X_k = x_1, \bX_{k^-}, \bH]\right] - \mathbb{E}\left[\mathbb{E}[X_j \mid X_k = x_0, \bX_{k^-}, \bH]\right].
\end{equation}

This definition concisely represents the total effect and aligns with the traditional concept of total effect for causal inference \citep{pearl2009causality, pearl2012calculus}, as detailed in the Appendix \ref{sec:ce_def}.

Since the unobserved confounders \(\bH\) in \eqref{eq:DAG_ate_conditionalH} generally render direct estimation infeasible without further assumptions, we transform \eqref{eq:DAG_ate_conditionalH} into an estimable quantity. Define residual variables \( S_l = X_l - \mathbb{E}[X_l \mid \bH] \) for each \( l = 1, \ldots, p \), and denote \( \bS = (S_1, \ldots, S_p)^\top \).

\begin{figure}[H]
	\centering
	\begin{subfigure}
		\centering
		\includegraphics[width=0.49\textwidth]{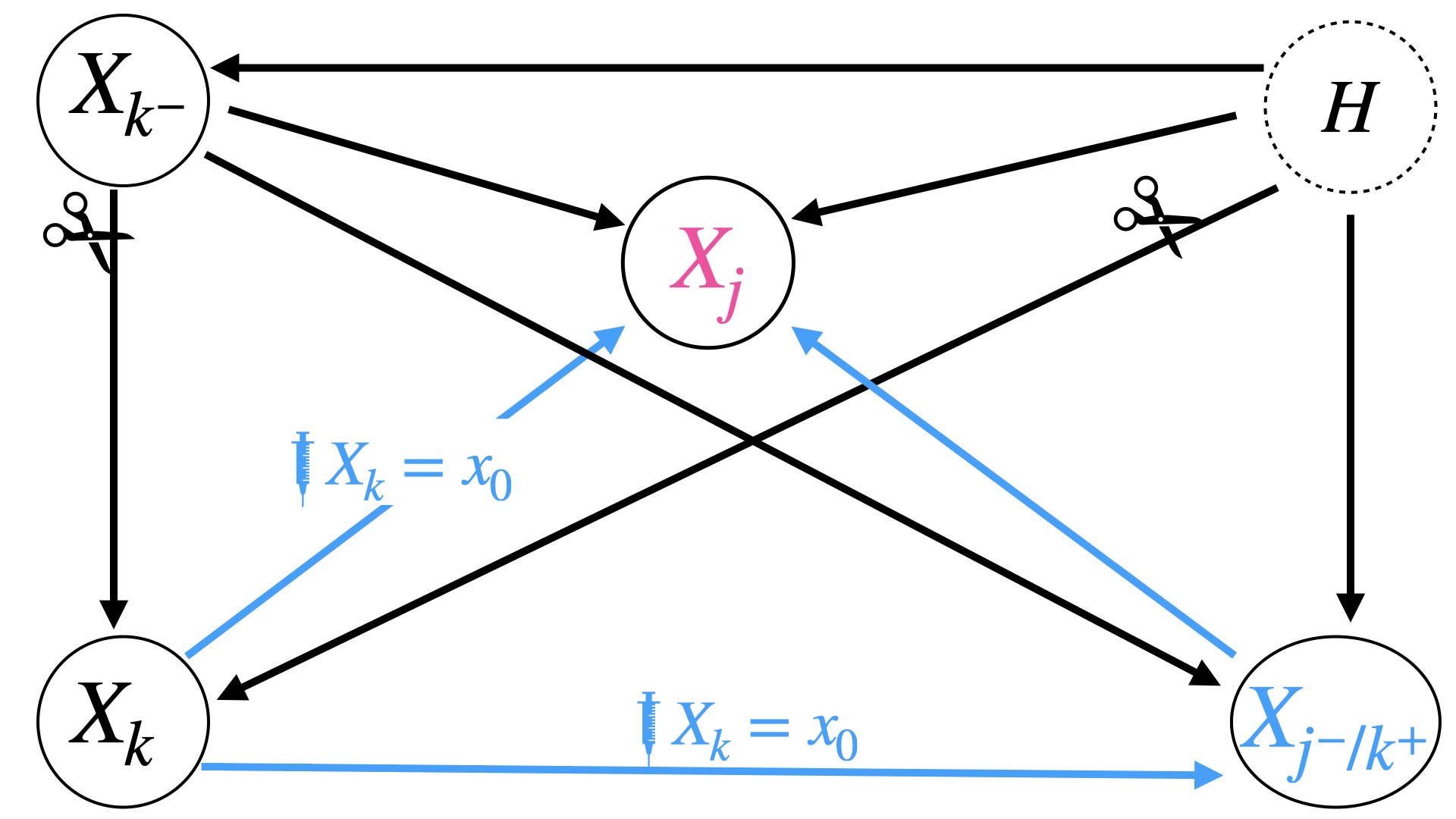}
	\end{subfigure}
	\hfill
	\begin{subfigure}
		\centering
		\includegraphics[width=0.49\textwidth]{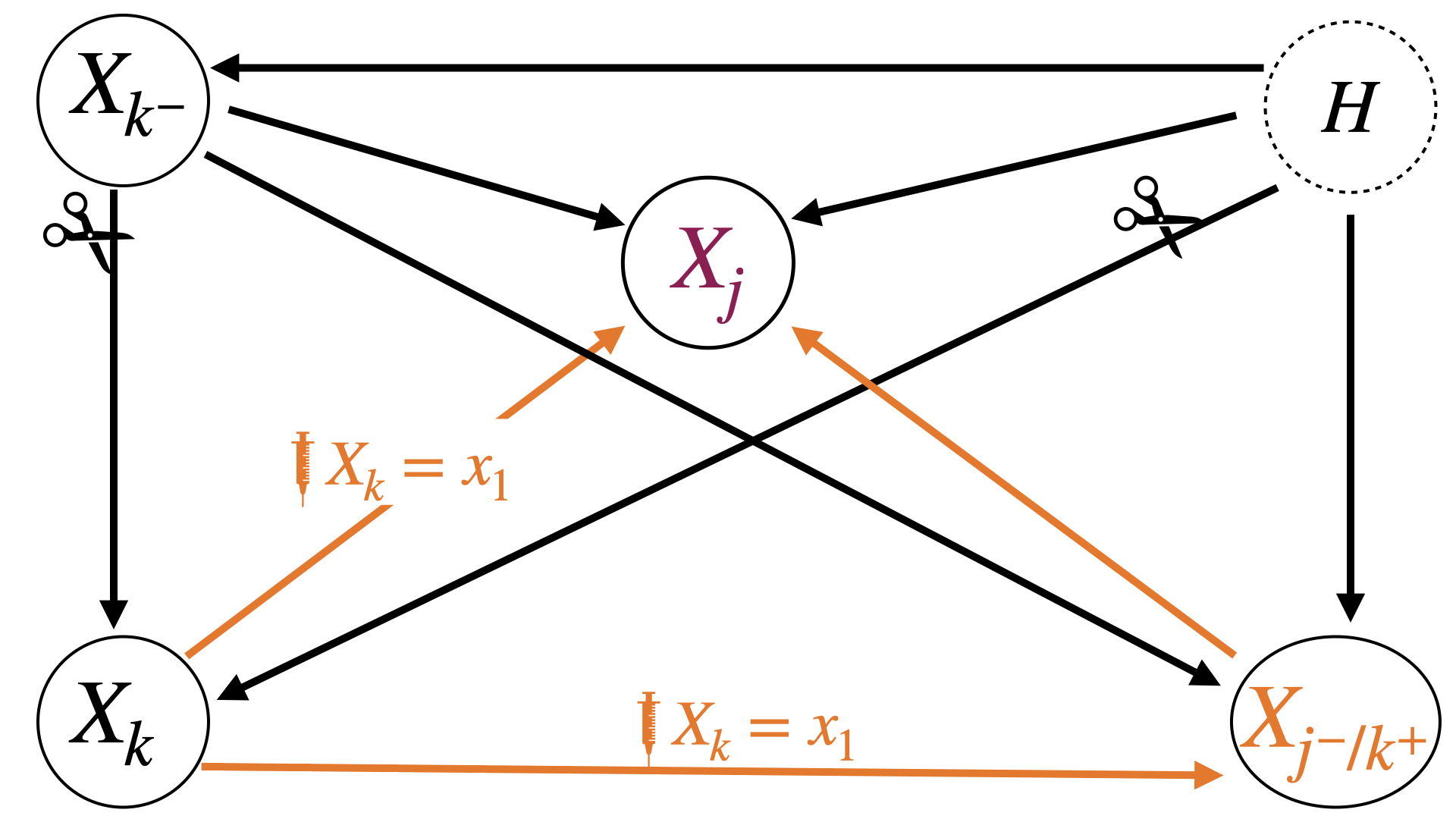}
	\end{subfigure}
	\caption{This figure illustrates the total causal effect of \(X_k\) on $X_j$ when \(X_k\) is changed from \(x_0\) (left) to \(x_1\) (right). By comparing the value of $X_j$ in both scenarios, we capture all causal pathways through which \(X_k\) exerts its influence on $X_j$.}
	\label{fig:total_effect_illustration}
\end{figure}

\begin{assumption}[Proxy Sufficiency for DAGs]
	\label{assu:proxy_DAG} 
	The residual \( X_j - S_j \) is conditionally independent of \( \bH \) given \( (\bS_{j^-}, \bX_{k^-}, X_k) \).
\end{assumption}

\begin{lemma}
	\label{lem:proxy_DAG}
	Under Assumption \ref{assu:proxy_DAG}, the total effect from \( X_k \) to \( X_j \) can be expressed as  
	\begin{equation}
		\label{eq:TE1}
		\begin{split}
			\tau(k,j;x_1,x_0) & = \E[\E[X_{j}|\bS,\bX_{k^-},X_{k} =x_1] - \E[X_{j}|\bS,\bX_{k^-},X_{k} =x_0]]
		\end{split}
	\end{equation}
\end{lemma}

Suppose we have an observed data matrix 
\(\mathbf{X}_{p\times n} = \bigl(\mathbf{X}_{\cdot,1},\ldots,\mathbf{X}_{\cdot,n}\bigr)\),
where each column \(\mathbf{X}_{\cdot,i}\) is an independent realization of the random vector \(\bX\), and let \(\mathbf{X}_{k^-,i}\) denote the subset of components indexed by \(k^-\) in the \(i\)-th observation.

Next, consider a conditional diffusion model that generates synthetic samples
\(\{\mathbf{X}_{j,i}^{(m)}(x_1)\}\) and \(\{\mathbf{X}_{j,i}^{(m)}(x_0)\}\) for \(i=1,\ldots,n\) and \(m=1,\ldots,M\). These samples satisfy
\[
\mathbf{X}_{j,i}^{(m)}(x_1) \sim \mathbb{P}
(X_j \,\bigl\lvert\, \bS=\widehat{\mathbf{S}}_{\cdot,i}, \,\bX_{k^-}=\mathbf{X}_{k^-,i},\, X_k=x_1),
\]
and analogously for \(x_0\). In other words, \(\mathbf{X}_{j,i}^{(m)}(x_1)\) (resp.\ \(\mathbf{X}_{j,i}^{(m)}(x_0)\)) is drawn from the conditional distribution of $X_j$ given \(\bS=\widehat{\mathbf{S}}_{\cdot,i}\), \(\bX_{k^-}=\mathbf{X}_{k^-,i}\), and \(X_k=x_1\) (resp.\ \(x_0\)).

By averaging over these synthetic draws, we approximate
\begin{equation*}
    \E\!\bigl[X_j \mid \bS=\widehat{\mathbf{S}}_{\cdot,i},\,\bX_{k^-}=\mathbf{X}_{k^-,i},\,X_k=x_1\bigr]
\end{equation*}
with $\frac{1}{M}\sum_{m=1}^M \mathbf{X}_{j,i}^{(m)}(x_1)$,
and similarly for \(x_0\). Consequently, the total causal effect of changing \(X_k\) from \(x_0\) to \(x_1\) on $X_j$ can be estimated by
\[
\widehat{\tau} 
\;=\; \frac{1}{M}\sum_{m=1}^M \mathbf{X}_{j,i}^{(m)}(x_1) 
\;-\; \frac{1}{M}\sum_{m=1}^M \mathbf{X}_{j,i}^{(m)}(x_0).
\]

\subsection{Connection between the total and direct effects over a DAG}

We now discuss the concept of direct effect \citep{robins1992identifiability}, defined as the expected change in \( X_j \) induced by changing \( X_k \) from \( x_0 \) to \( x_1 \) while keeping all mediating variables, \(\bX_{j^-/k^+}\), fixed at whatever value they would have obtained under \( X_k = x_0 \). The direct effect can be expressed by such an expression:
\begin{equation}
    \label{eq:DE}
    \begin{split}
        \tau_{\operatorname{d}}(k,j;x_1,x_0) & = \E_{\bX_{k^-},\bH}[\E_{\bX_{j^-/k^+}}[\E[X_{j}|X_{k}=x_1,\bX_{j^-/k^+},\bX_{k^-},\bH]|X_{k}=x_0,\bX_{k^-},\bH]] \\ 
        & - \E_{\bX_{k^-},\bH}[\E_{\bX_{j^-/k^+}}[\E[X_{j}|X_{k}=x_0,\bX_{j^-/k^+},\bX_{k^-},\bH]|X_{k}=x_0,\bX_{k^-},\bH]].
    \end{split}
\end{equation}

The direct effect considers only the pathway in which exposure directly affects the outcome, excluding any mediated pathways. Proposition \ref{prop:de_edge} establishes a connection between the direct effect and the presence of directed edges in a DAG.

\begin{figure}
	\centering
	\begin{subfigure}
		\centering
		\includegraphics[width=0.49\textwidth]{Figures/effect_x0.png}
	\end{subfigure}
	\hfill
	\begin{subfigure}
		\centering
		\includegraphics[width=0.49\textwidth]{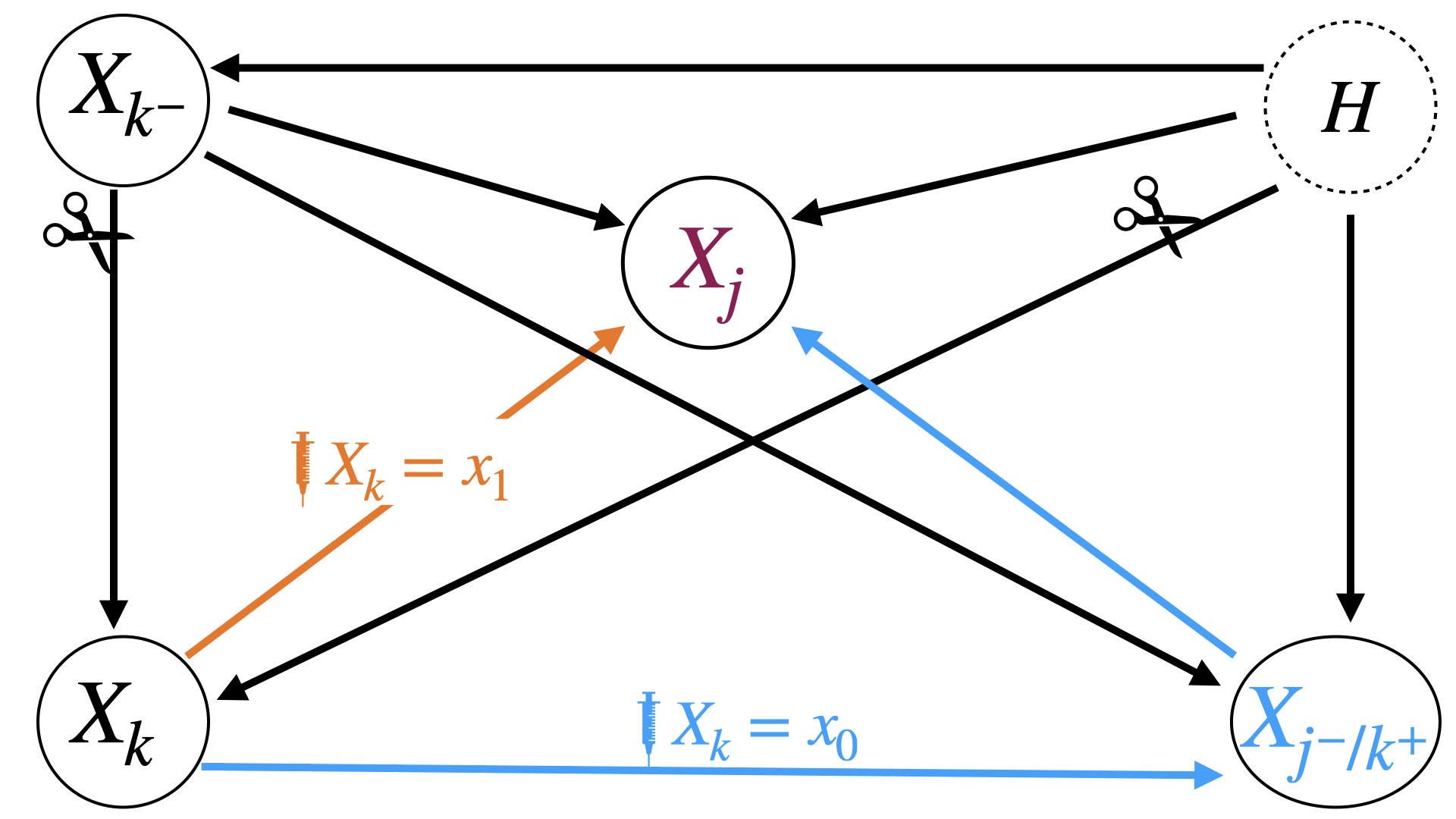}
	\end{subfigure}
	\caption{An illustration of computing the direct effect from $X_k$ to $X_j$ when $X_k$ is changed from $x_0$ to $x_1$. We compute the difference of $X_j$ in the right and left panel when we change $X_k$ from $x_0$ to $x_1$ while keeping all mediating variables $X_{j^-/k^+}$ constant at whatever value they would have obtained under $X_k=x_0$.}
	\label{fig:direct_effect_illustration}
\end{figure}

\begin{proposition} \label{prop:de_edge}
    Given a causal order \( \pi \), the absence of a directed edge from \( X_k \) to \( X_j \) indicates a zero direct effect from \( X_k \) to \( X_j \), that is:
    \begin{equation}
        \begin{split}
            (k \to j) \not \in \mathcal{E}^* \Rightarrow \tau_{\operatorname{d}}(k, j; x_1, x_0) = 0 \quad \forall x_1, x_0 \in \mathbb{R}.
        \end{split}
    \end{equation}
\end{proposition}

However, a zero direct effect from \( X_k \) to \( X_j \) does not necessarily indicate the absence of a directed edge from \( X_k \) to \( X_j \) without further restrictions. We illustrate this with the following example.

\begin{example}
    For variables \( X_1, X_2 \) with causal order \( (1, 2) \) and the following structural equation model:
    \begin{equation*}
        \begin{split}
            X_1  = \varepsilon_1, \quad X_2  = X_1 \varepsilon_2,
        \end{split}
    \end{equation*}
    where \( \varepsilon_1, \varepsilon_2 \) are independent standard Gaussian noise. By Assumption \ref{ass:strong_ignorability_DAG}, \( X_1 \) is a parent of \( X_2 \), while \( \tau_{\operatorname{d}}(1, 2; x_1, x_0) = 0 \) for any \( x_0, x_1 \in \mathbb{R} \).
\end{example}

In contrast, the total causal effect represents the overall impact of an exposure on an outcome, accounting for all possible pathways through which the exposure exerts its influence. Although direct effects are valuable for designing interventions that target specific pathways, their calculation can be challenging without a well-specified causal model, as it requires careful adjustment for mediators along the pathway. Furthermore, interpreting direct effects requires caution, as mediators can be unknown, difficult to measure, or too numerous to account for comprehensively. In comparison, total causal effects provide a broader perspective, capturing both direct and indirect influences. This makes them particularly useful in systems where multiple variables interact in complex and interdependent ways.

We provide a simple example below to illustrate the difference between our defined total effect in \eqref{eq:DAG_ate_conditionalH} and the direct effects.

\begin{example}
    To investigate causal effects among four variables \( X_1, \ldots, X_4 \), consider a model of structural equations for \( \bX = (X_1, \ldots, X_4) \): \begin{equation} \label{ex} \bX = \boldsymbol{V}^{\top} \bX + \boldsymbol{\varepsilon},
    \end{equation}
    where \( \boldsymbol{\varepsilon} \) is standard Gaussian noise and \( V_{kj} = 0 \) for \( k \geq j \). In \eqref{ex}, a causal order is \( (1, 2, 3, 4) \), and \( V_{kj} \) represents the direct effect of \( X_k \to X_j \). 

    The total effect from \( X_2 \) to \( X_4 \) can be calculated as: 
$\tau (2, 4; 1, 0) = V_{24} + V_{23}V_{34}$,
and the direct effect from \( X_2 \) to \( X_4 \) as:
        $\tau_{\operatorname{d}}(2, 4; 1, 0) = V_{24}$.
\end{example}

\section{Theoretical Properties of ACEE}

In this section, we analyze the theoretical properties of the ACEE estimator. Consider a target dataset 
$\mathcal{D} = \{\mathbf{X}_{\cdot,i}\}_{i=1}^n$,
sampled independently from a distribution $\mathbb{P}$, and an independent source dataset 
$\mathcal{D}_s = \{\mathbf{X}^s_{\cdot,i}\}_{i=1}^{n_s}$,
sampled from a distribution $\mathbb{Q}$. We assume that there exists a shared latent embedding, denoted by $h^0$, such that
\[
\mathbb{P}(X_j \mid \bX_{k^+}, \bS) = \mathbb{P}(X_j \mid h^0(\bX_{k^+}, \bS)), \quad 
\mathbb{Q}(X_j \mid \bX_{k^+}, \bS) = \mathbb{Q}(X_j \mid h^0(\bX_{k^+}, \bS)).
\]

To estimate the shared latent embedding, we first minimize the empirical loss on the source dataset along with its estimated proxy $\{\widehat{\mathbf{S}}_{\cdot,i}^{s}\}_{i=1}^{n_s}$:
$$(\widehat{\theta}_s,\widehat{h}) = \argmin_{\theta_s \in \mathcal{F}, \, h \in \Theta_h} \sum_{i=1}^{n_s} \ell_s \Bigl(\mathbf{X}_{j,i}^s, \mathbf{X}_{k^+,i}^s, \widehat{\mathbf{S}}_{\cdot,i}^{s}; \theta_s, h\Bigr),$$
where $\ell_s$ is defined in \eqref{eq:loss_score_match_implement_source}, $\Theta_h$ denotes the parameter space to estimate $h$ and $\mathcal{F}$ refers to the ReLU neural networks detailed in Appendix \ref{sec:cdm}.

Given the estimated latent embedding $\widehat{h}$, we then minimize the empirical loss on the target dataset and its estimated proxy $\{\widehat{\mathbf{S}}_{\cdot,i}\}_{i=1}^{n}$:
    $$\widehat{\theta} = \argmin_{\theta \in \mathcal{F}} \sum_{i=1}^{n} \ell \Bigl(\mathbf{X}_{j,i}, \mathbf{X}_{k^+,i}, \widehat{\mathbf{S}}_{\cdot,i}; \theta, \widehat{h}\Bigr),$$
where $\ell$ is defined in \eqref{eq:loss_score_match_implement_target}.

To analyze the estimation accuracy of the total causal effects, we impose the following technical assumptions.

\begin{assumption}[Transferability via Conditional Models]
    \label{assu:transferability}
    There exists a constant $c_1 > 0$ such that for any $h \in \Theta_h$, 
    \[
    |\delta(h) - \delta(h^0)| \leq c_1 \, |\delta_s(h) - \delta_s(h^0)|,
    \]
    where 
    \[
    \delta(h) = \inf_{\theta \in \mathcal{F}} \mathbb{E}_{(\bX,\bS)} \Bigl[\ell(X_j, \bX_{k^+}, \widehat{\bS}; \theta, h) - \ell(X_j, \bX_{k^+}, \bS; \theta^0, h^0)\Bigr]; 
    \]
    \[
    \delta_s(h) = \inf_{\theta_s \in \mathcal{F}} \mathbb{E}_{(\bX^s,\bS^s)} \Bigl[\ell_s(X_j^s, \bX_{k^+}^s, \widehat{\bS}^s; \theta_s, h) - \ell_s(X_j^s, \bX_{k^+}^s, \bS^s; \theta_s^0, h^0)\Bigr].
    \]
\end{assumption}

Assumption \ref{assu:transferability} quantifies how the excess risk associated with the latent representation transfers from the source dataset to the target dataset. Define the source excess risk by
\[
\rho_s^2\left(\gamma_s^0, \hat{\gamma}_s\right) = \mathbb{E}_{(\bX^s,\bS^s)} \Bigl[\ell_s(X_j^s, \bX_{k^+}^s, \widehat{\bS}^s; \theta_s, h) - \ell_s(X_j^s, \bX_{k^+}^s, \bS^s; \theta_s^0, h^0)\Bigr],
\]
where $\gamma_s = (\theta_s, h)$.

The next assumption concerns a specific generation error in the source generator.

\begin{assumption}[Source Error]
    \label{assu:source_error}
    There exists a sequence $\epsilon_s$ (indexed by $n_s$) such that for any $\varepsilon \geq \epsilon_s$, 
    \[
    P\Bigl(\rho_s\left(\gamma_s^0, \hat{\gamma}_s\right) \geq \varepsilon\Bigr) \leq \exp \Bigl(-c_2 \, n_s^{1-\xi} \, \varepsilon^2\Bigr),
    \]
    where $c_2 > 0$ and $\xi > 0$ are constants, and $n_s^{1-\xi} \, \epsilon_s^2 \rightarrow \infty$ as $n_s \rightarrow \infty$.
\end{assumption}

\begin{definition}[H\"older Ball]
	\label{def:holder_ball}
	Let $\beta = \lfloor \beta \rfloor + \gamma > 0$, with $\gamma \in [0,1)$. For a function $f: \mathbb{R}^d \rightarrow \mathbb{R}$, the H\"older ball of radius $B>0$ is defined as
	\[
	\mathcal{H}^{\beta}\left(\mathbb{R}^{d}, B\right) = \Bigl\{ f: \mathbb{R}^{d} \rightarrow \mathbb{R} \; \Bigl| \; \|f\|_{\mathcal{H}^{\beta}\left(\mathbb{R}^{d}\right)} < B \Bigr\},
	\]
	where the H\"older norm is given by
	\[
	\|f\|_{\mathcal{H}^{\beta}\left(\mathbb{R}^{d}\right)} := \max_{\|\mathbf{s}\|_{1} < \lfloor \beta \rfloor} \sup_{\mathbf{x}} \left|\partial^{\mathbf{s}} f(\mathbf{x})\right| + \max_{\|\mathbf{s}\|_{1} = \lfloor \beta \rfloor} \sup_{\mathbf{x} \neq \mathbf{z}} \frac{\left|\partial^{\mathbf{s}} f(\mathbf{x}) - \partial^{\mathbf{s}} f(\mathbf{z})\right|}{\|\mathbf{x}-\mathbf{z}\|_{\infty}^{\gamma}}.
	\]
\end{definition}

Focusing on smooth distributions, we impose the following condition on the target conditional density.

\begin{assumption}[Conditional Density]
	\label{assu:target_density}
	The true conditional density $\mathbb{P}(z \mid \mathbf{x},\mathbf{s})$ of $X_j$ given $\bX_{k^+} = \mathbf{x}$ and $\bS = \mathbf{s}$ is assumed to be of the form
$\mathbb{P}(z \mid \mathbf{x},\mathbf{s}) = \exp\bigl(-C \|z\|_2^2\bigr) f\Bigl(z, h^0(\mathbf{x},\mathbf{s})\Bigr)$,
	where $C > 0$ is a constant, and $f$ is a nonnegative function that is bounded away from zero and belongs to the H\"older ball $\mathcal{H}^\beta\bigl(\mathbb{R} \times [0,1]^{d_h},B\bigr)$ with radius $B>0$ and smoothness degree $\beta>0$. Here, $d_h$ denotes the dimension of the latent embedding $h^0(\mathbf{x},\mathbf{s})$. 
\end{assumption}
Assumption \ref{assu:target_density} essentially requires that the density ratio between the target density and a Gaussian kernel lies within a H\"older class. Although we assume $h^0(\mathbf{x},\mathbf{s})$ is bounded for simplicity, our analysis can be extended to the unbounded, light-tail case.

\begin{assumption}[Proxy Estimation]
	\label{assu:proxy_estimation_DAG}
	Denote $\widehat{\eta}(\bX_{k^-},\bS) = \E_{\widehat{\mathbb{P}}_n}[X_j \mid \bX_{k^-}, X_k = x_1, \bS] - \E_{\widehat{\mathbb{P}}_n}[X_j \mid \bX_{k^-}, X_k = x_0, \bS] $, we assume that 
	\begin{align*}
		\E \left| \frac{1}{n} \sum_{i=1}^n \widehat{\eta}(\mathbf{X}_{k^-,i},\widehat{\mathbf{S}}_{\cdot,i}) - \E_n[\widehat{\eta}(\mathbf{X}_{k^-,1},\mathbf{S}_{\cdot,1})] \right| = O(\frac{1}{\sqrt{n}})
	\end{align*}
	where $\widehat{\mathbb{P}}_n$ denotes the learned conditional probability by the diffusion model from the source and target dataset.
\end{assumption}

\begin{theorem}[Causal effect estimation] 
    \label{thm:acee_rate}
    Under Assumptions \ref{assu:proxy_DAG}--\ref{assu:proxy_estimation_DAG}, the mean absolute error of the causal effect estimator satisfies
    \[
    \mathbb{E}\Bigl[ \Bigl|\widehat{\tau}(k,j,x_1,x_0)-\tau(k,j,x_1,x_0)\Bigr| \Bigr] = O\Biggl(n^{-\frac{\beta}{1+d_{h}+2\beta}} \, \bigl(\log (n)\bigr)^{\max\{11,\,(\beta+5)/2\}} + \frac{1}{\sqrt{M}} + \epsilon_s\Biggr).
    \]
\end{theorem}

The three components in the error bound can be interpreted as follows: 

\begin{enumerate} \item The error term $n^{-\frac{\beta}{1 + d_{h} + 2\beta}}\left( \log(n)\right)^{\max\{11,(\beta + 5)/2\}}$ reflects the generation error arising from estimating the conditional density (cf. \cite{fu2024unveil}).

\item The term $\frac{1}{\sqrt{M}}$ corresponds to the Monte Carlo error due to the use of synthetic samples, which can be made arbitrarily small by increasing \( M \).

\item The term $\epsilon_s$ represents the generation error from the source data. Since \( n_s \gg n \) is typical—owing to the use of a large pretrained model—this error is generally negligible, and the first term tends to dominate the overall estimation accuracy (cf. \cite{tian2024enhancing}).

\end{enumerate}

The remainder of this section focuses on the estimation accuracy for the bias-corrected estimator described in Section~\ref{bias-correction}.

\begin{assumption}[Model Assumptions]
\label{assu:data}
We impose the following regularity conditions:
\begin{enumerate}
    \item Overlap: There exists a constant $\eta > 0$ such that 
    $\eta < \mathbb{P}(D=1 \mid \bX, \bS) < 1 - \eta$.
    
    \item Response Function: For $d \in \{0,1\}$, the response function has a finite second moment, i.e.,
$\mathbb{E}\left[\mu_d^2(\bX, \bS)\right] < \infty$,
    where $\mu_d(\bX, \bS) = \mathbb{E}\left[Y \mid \bX, \bS, D=d\right]$.
    
    \item Residuals: For $d \in \{0,1\}$, the conditional second moment of the residuals 
$U_d = Y(d) - \mu_d(\bX, \bS)$
    is uniformly bounded almost surely over $(\bX,\bS)$.
\end{enumerate}
\end{assumption}

The first condition is the standard overlap requirement, while the remaining conditions impose conventional regularity constraints.

\begin{assumption}[Proxy Estimation Consistency]
\label{assu:uncorrected_estimator}
For each $d \in \{0,1\}$, assume there exists a deterministic function $\bar{\mu}_d(\bX, \bS)$ satisfying
$\mathbb{E}\left[\bar{\mu}_d^2(\bX, \bS)\right] < \infty$,
such that the estimated outcome model $\widehat{\mu}_d(\bX_i, \widehat{\bS}_i)$ satisfies:
\[
\max_{1 \le i \le n} \left|\widehat{\mu}_d(\bX_i, \widehat{\bS}_i) - \bar{\mu}_d(\bX_i, \bS_i)\right| \xrightarrow{p} 0, \quad 
\max_{1 \le i \le n} \left\|(\bX_i, \widehat{\bS}_i) - (\bX_i, \bS_i)\right\| \xrightarrow{p} 0.
\]
\end{assumption}

Assumption \ref{assu:uncorrected_estimator} permits misspecification of the outcome model by requiring only estimation consistency.

\begin{assumption}[Bounded Neighborhood Density]
	\label{assu:density_neighbor}
	For observation $1$ with $(\bX_1, \bS_1)$, we assume that $\mu_d$ and $\bar{\mu}_d$ are continuous at $(\bX_1, \bS_1)$; $d=0,1$, and if there exists a radius $r > 0$ such that the density of $(\bX, \bS)$ is bounded below by a positive constant if $\|(\bX,\bS) - (\bX_1,\bS_1)\| < r$ for both treated and control groups.
\end{assumption}

\begin{theorem}[Estimation Consistency]
\label{thm:consistency_bad_generation}
Under Assumptions~\ref{assu:conditional_random}, \ref{assu:proxy}, \ref{assu:data}, and \ref{assu:uncorrected_estimator}, and as $n, M \rightarrow \infty$, if we choose $N \rightarrow \infty$ but $\frac{N \log(n)}{n} \rightarrow 0$, then
the bias-corrected estimator satisfies:
\[
\widehat{\tau}^{\mathrm{c}} - \tau \xrightarrow{p} 0,
\]
where $\tau$ denotes the average treatment effect (ATE), and $n$ and $M$ are the sample and Monte Carlo sizes.
For the individual treatment effect (ITE) at observation $1$, if Assumption \ref{assu:density_neighbor} is satisfied, then
\[
\widehat{\tau}^{\mathrm{c}}(\bX_1, \widehat{\bS}_1) - \tau(\bX_1, \bS_1) \xrightarrow{p} 0.
\]
\end{theorem}

Theorem \ref{thm:consistency_bad_generation} demonstrates the consistency of the bias-corrected estimators for both the average and individual treatment effects, even under potential misspecification of the outcome model.

\section{Numerical Examples}

\subsection{Simulations}

This subsection conducts simulation studies to evaluate the finite sample performance of the proposed ACEE method. We also compare it with state-of-the-art methods for average treatment effect estimation, including the bias-corrected nearest-neighbor matching estimator \citep{lin2023estimation}, the empirical balancing calibration weighting estimator \citep{chan2016globally}, and causal forest \citep{wager2018estimation}.

Consider the following models in the simulation studies.

\begin{enumerate}
\item (M1). A nonlinear model with an additive error term: $Y = X_1^2 + X_1 X_2 + \exp(X_3+D) + \sin(X_4X_5) + \varepsilon, \varepsilon \sim N(0,1)$
\item (M2). A model with an additive error term whose variance depends on the predictors: $Y = X_1^2 + \exp(X_3+D) + \sin (X_4 X_5) + (10D+X_5^2/2) \times \varepsilon, \varepsilon \sim N(0,1)$.
\item (M3). A model with a multiplicative non-Gaussian error term: $Y = (X_1^2 + \sin(X_2 X_3) + D)\times \exp(\varepsilon), \varepsilon \sim N(0,1)$.
\item (M4). A model with an unobserved confounder. $Y = X_1^2 + D X_1 + \sin(X_4 X_5) + U^2 + \varepsilon, \varepsilon \sim N(0,1)$
\end{enumerate}

In (M1)-(M3), the covariate vector $X$ is generated from standard multivariate normal distribution. We determine the treatment by $P(D=1|\bX)= 0.1 + \frac{0.8}{1+\exp(X_1X_2)}$. In (M4), we set $X_1 = 1.5 U + \varepsilon_1, X_2 = U + \varepsilon_2. X_3 = U + \varepsilon_3$ , where $(U,\varepsilon_1,\varepsilon_2,\varepsilon_3)$ is generated from standard multivariate normal distribution. We determine the treatment by $P(D=1|\bX,U)= 0.1 + \frac{0.8}{1+\exp(-U)}$.

\begin{table}[H]
\footnotesize
    \centering
    \begin{tabular}{c c c c c c}
    \hline
      Model   &  Sample Size  & ACEE & EBCW & CF & BCNNM\\
    \hline
     M1  & 50 & 0.040 & 0.035 & 0.081 & 0.067\\
         & 100 & 0.040 & 0.019 & 0.030 & 0.028\\
         & 200 & 0.033 & 0.011 & 0.010 & 0.022\\
     M2  & 50 & 0.105 & 8.858 & 8.978 & 7.660\\
         & 100 & 0.093 & 5.205 & 5.318 & 4.630\\
         & 200 & 0.038 & 1.551 & 1.349 & 1.312\\
     M3  & 50 & 0.082 & 9.632 & 13.086 & 5.426\\
         & 100 & 0.075 & 1.276 & 0.897 & 0.773\\
         & 200 & 0.017 & 0.736 & 0.559 & 0.763\\
     M4  & 50 & 0.039 & 1.455 & 1.188 & 0.812\\
         & 100 & 0.019 & 1.992 & 0.302 & 0.564\\
         & 200 & 0.015 & 0.581 & 0.071 & 0.209\\
     \hline
    \end{tabular}
    \caption{Mean squared error (MSE) of the estimated average treatment effect (ATE) over 10 simulation 
repetitions.}
    \label{tab:ate1}
\end{table}
As shown in Table \ref{tab:ate1}, ACEE has significant advantages over competitors, especially when the model comes with heteroscedastic noise and the sample size is small.

We also examine how bias-correction mechanism helps when the distribution of auxiliary data deviates from the distribution of target data. Let $\text{M}$ represent the structural equation model of the data, for the auxiliary data, we set $P(\text{M} = \text{M1}) = \eta, P(\text{M}=\text{M3}) = 1-\eta$, and for the target data, we set $P(\text{M}=\text{M1})=1$.  

\begin{table}[H]
\footnotesize
    \centering
    \begin{tabular}{c c c c c c}
    \hline
     $\eta $ &  0.2 & 0.4 & 0.6 & 0.8\\
    \hline
     ACEE    & 0.181 & 0.261 & 0.248 & 0.253\\
     ACEE-bc & 0.026 & 0.014 & 0.020 & 0.042 \\
     \hline
    \end{tabular}
    \caption{Mean squared error (MSE) of the estimated average treatment effect (ATE) over 10 simulation repetitions.}
    \label{tab:ate2}
\end{table}
As shown in Table \ref{tab:ate2}, the bias correction mechanism helps when the distribution of auxiliary data deviates from the distribution of target data.

\subsection{CSuite Dataset}
CSuite contains a collection of synthetic datasets designed to benchmark causal machine learning algorithms, originally introduced in \cite{geffner2024deep}. We selected four CSuite datasets for our numerical studies:
\begin{enumerate}
    \item (\textit{Nonlinear Simpson}) A demonstration of Simpson's paradox using a continuous structural equation model (SEM), where $\operatorname{Cov}(X_1, X_2)$ has the opposite sign of $\operatorname{Cov}(X_1, X_2 | X_0)$.
    \item (\textit{Symprod Simpson}) A multimodal dataset where $\operatorname{Cov}(X_0, X_2) = 0$ and $\operatorname{Cov}(X_1, X_2) = 0$, highlighting the importance of nonlinear function estimation.
    \item (\textit{Large Backdoor}) A nonlinear SEM with non-Gaussian noise, where adjusting for all confounders is valid but leads to high variance.
    \item (\textit{Weak Arrows}) A variant of \textit{Large Backdoor} with numerous additional edges.
\end{enumerate}

 Each dataset includes a training set of 2,000 samples and an intervention set. Each intervention test consists of a treatment variable, a treatment value, a reference treatment value, and an effect variable. The ground-truth average treatment effect (ATE) is estimated using the treated and reference intervention test sets.
  
  As a benchmark, we include Deep End-to-End Causal Inference (DECI) \citep{geffner2024deep}, a state-of-the-art method that jointly performs causal discovery and causal effect estimation in nonlinear settings by leveraging variational inference and flow-based models. To ensure a fair comparison, DECI is provided with the same partial causal order information utilized by ACEE, which serves as constraints for its causal discovery process.

  As suggested by Table \ref{tab:Csuite1}, ACEE achieves comparable performance to DECI across all four data examples, despite ACEE not relying on auxiliary samples to enhance the diffusion model training. Notably, the random error in these examples is additive, a setting that aligns with DECI’s assumptions. Since DECI explicitly assumes additive noise, it has a structural advantage over ACEE, which does not impose this assumption.

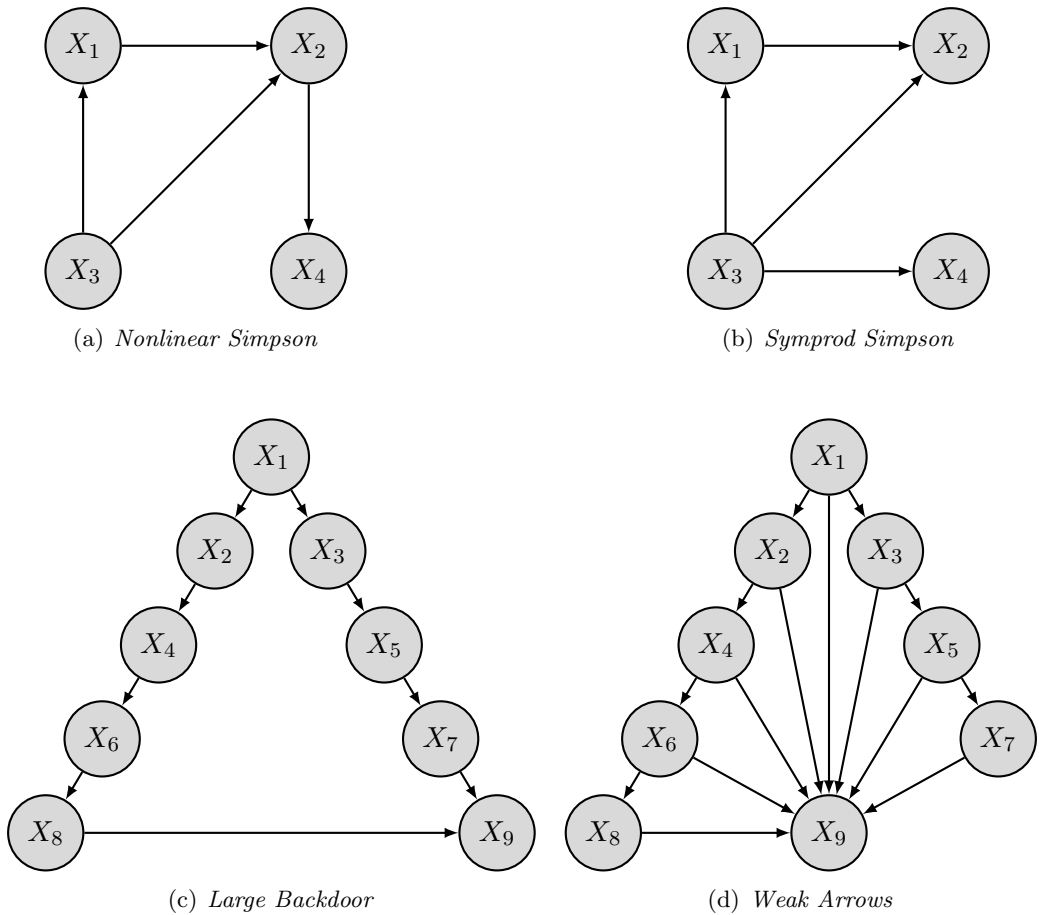
\begin{figure}[H]
\centering
\subfigure[\textit{Nonlinear Simpson}]{
\begin{tikzpicture}[
    roundnode/.style={circle, draw=black, fill=gray!30, thick, minimum size=10mm},
    arrow/.style={-latex, thick},
    node distance=1.5cm and 2cm
]

\node[roundnode] (X1) at (0,6) {$X_1$};
\node[roundnode] (X2) at (3,6) {$X_2$};
\node[roundnode] (X3) at (0,3) {$X_3$};
\node[roundnode] (X4) at (3,3) {$X_4$};

\draw[arrow] (X1) -- (X2);
\draw[arrow] (X3) -- (X1);
\draw[arrow] (X3) -- (X2);
\draw[arrow] (X2) -- (X4);
\end{tikzpicture}
}
\hspace{4cm}
\subfigure[\textit{Symprod Simpson}]{
\begin{tikzpicture}[
    roundnode/.style={circle, draw=black, fill=gray!30, thick, minimum size=10mm},
    arrow/.style={-latex, thick},
    node distance=1.5cm and 2cm
]

\node[roundnode] (X1) at (0,6) {$X_1$};
\node[roundnode] (X2) at (3,6) {$X_2$};
\node[roundnode] (X3) at (0,3) {$X_3$};
\node[roundnode] (X4) at (3,3) {$X_4$};

\draw[arrow] (X1) -- (X2);
\draw[arrow] (X3) -- (X1);
\draw[arrow] (X3) -- (X2);
\draw[arrow] (X3) -- (X4);

\end{tikzpicture}
}
\vspace{0.5cm} 

\subfigure[\textit{Large Backdoor}]{
\begin{tikzpicture}[
    roundnode/.style={circle, draw=black, fill=gray!30, thick, minimum size=10mm},
    arrow/.style={-latex, thick},
    node distance=1.5cm and 2cm
]

\node[roundnode] (X1) at (0, 3) {$X_1$};
\node[roundnode] (X2) at (-0.75, 1.75) {$X_2$};
\node[roundnode] (X3) at (0.75, 1.75) {$X_3$};
\node[roundnode] (X4) at (-1.5, 0.5) {$X_4$};
\node[roundnode] (X5) at (1.5, 0.5) {$X_5$};
\node[roundnode] (X6) at (-2.25, -0.75) {$X_6$};
\node[roundnode] (X7) at (2.25, -0.75) {$X_7$};
\node[roundnode] (X8) at (-3, -2) {$X_8$};
\node[roundnode] (X9) at (3, -2) {$X_9$};

\draw[arrow] (X1) -- (X2);
\draw[arrow] (X1) -- (X3);
\draw[arrow] (X2) -- (X4);
\draw[arrow] (X3) -- (X5);
\draw[arrow] (X4) -- (X6);
\draw[arrow] (X6) -- (X8);
\draw[arrow] (X8) -- (X9);
\draw[arrow] (X5) -- (X7);
\draw[arrow] (X7) -- (X9);
\end{tikzpicture}
}
\subfigure[\textit{Weak Arrows}]{
\begin{tikzpicture}[
    roundnode/.style={circle, draw=black, fill=gray!30, thick, minimum size=10mm},
    arrow/.style={-latex, thick},
    node distance=1.5cm and 2cm
]

\node[roundnode] (X1) at (0, 3) {$X_1$};
\node[roundnode] (X2) at (-0.75, 1.75) {$X_2$};
\node[roundnode] (X3) at (0.75, 1.75) {$X_3$};
\node[roundnode] (X4) at (-1.5, 0.5) {$X_4$};
\node[roundnode] (X5) at (1.5, 0.5) {$X_5$};
\node[roundnode] (X6) at (-2.25, -0.75) {$X_6$};
\node[roundnode] (X7) at (2.25, -0.75) {$X_7$};
\node[roundnode] (X8) at (-3, -2) {$X_8$};
\node[roundnode] (X9) at (0, -2) {$X_9$};

\draw[arrow] (X1) -- (X2);
\draw[arrow] (X1) -- (X3);
\draw[arrow] (X2) -- (X4);
\draw[arrow] (X3) -- (X5);
\draw[arrow] (X4) -- (X6);
\draw[arrow] (X6) -- (X8);
\draw[arrow] (X8) -- (X9);
\draw[arrow] (X5) -- (X7);
\draw[arrow] (X7) -- (X9);
\draw[arrow] (X1) -- (X9);
\draw[arrow] (X2) -- (X9);
\draw[arrow] (X3) -- (X9);
\draw[arrow] (X4) -- (X9);
\draw[arrow] (X5) -- (X9);
\draw[arrow] (X6) -- (X9);
\end{tikzpicture}
}

\caption{Parent-Child relationships of four selected Csuite datasets.}
\label{fig:four_dags}
\end{figure}

\begin{table}[H]
\footnotesize
    \centering
    \begin{tabular}{l c c c c}
    \hline
    Method & ACEE & DECI Gaussian & DECI Spline \\
    \hline
    \textit{Nonlinear Simpson} & 2.072 & 1.974 & 1.732 \\
    \textit{Symprod Simpson} & 0.186 & 0.326 & 0.368 \\
    \textit{Large Backdoor} & 0.411 & 0.196 & 0.193 \\
    \textit{Weak Arrows} & 0.49 & 0.212 & 0.263  \\
    \hline
    \end{tabular}
\caption{Root Mean Squared Errors (RMSEs) of the estimates and the estimated ground truth over 20 random seeds across four CSuite datasets. DECI Gaussian/Spline \cite{geffner2024deep} correspond to DECI using a Gaussian or spline noise model.}
    \label{tab:Csuite1}
\end{table}
To further assess the robustness of ACEE in scenarios where additive noise assumptions do not hold, we modify the \textit{Nonlinear Simpson} and \textit{Symprod Simpson} datasets by replacing the additive noise terms with multiplicative noise while preserving the parent-child relationships. The RMSE results for these modified datasets are presented in Table \ref{tab:Csuite_prod}. In these more challenging settings, where DECI’s additive noise assumption is violated, ACEE significantly outperforms DECI, demonstrating its superior adaptability to more flexible data-generating processes.

\begin{table}[H]
\footnotesize
    \centering
    \begin{tabular}{l c c c c}
    \hline
    Method & ACEE & DECI Gaussian & DECI Spline \\
    \hline
    \textit{Nonlinear Simpson (Multiplicative Error Term)} & 0.369 & 2.542 & 2.931 \\
    \textit{Symprod Simpson (Multiplicative Error Term)} & 0.074 & 0.689 & 0.201 \\
    \hline
    \end{tabular}
\caption{Root Mean Squared Errors (RMSEs) of the estimates and the ground truth over 20 random seeds across two modified CSuite datasets. DECI Gaussian/Spline \citep{geffner2024deep} correspond to DECI using a Gaussian or spline noise model.}
    \label{tab:Csuite_prod}
\end{table}

To demonstrate the effectiveness of transfer learning in improving causal effect estimation, we augment the training of diffusion models by incorporating auxiliary data for the \textit{Nonlinear Simpson} and \textit{Symprod Simpson} datasets. The auxiliary datasets preserve the same parent-child relationships as the original datasets but differ in their underlying distribution; see the appendix for further details on the data generation process. We evaluate the impact of auxiliary sample size on estimation accuracy by computing the RMSEs between the estimated and ground-truth average treatment effects (ATEs). As shown in Table \ref{tab:aux_sample_size}, ACEE exhibits a clear improvement in estimation accuracy as the number of auxiliary samples increases.

\begin{table}[H]
\footnotesize
    \centering
    \begin{tabular}{l c c c c c}
    \hline
    Auxiliary Sample Size & 2000 & 4000 & 6000 & 8000 & 10000 \\
    \hline
    nonlin\_simpson & 0.329 & 0.269 & 0.209 & 0.193 & 0.166 \\
    symprod\_simpson & 0.083 & 0.078 & 0.078 & 0.071 & 0.071 \\
    \hline
    \end{tabular}
    \caption{Root Mean Squared Errors (RMSEs) of ATE estimates over 20 random seeds across different auxiliary sample sizes.}
    \label{tab:aux_sample_size}
\end{table}

\subsection{IHDP Data Analysis}

In this subsection, we apply ACEE to the Infant Health and Development Program (IHDP) dataset introduced in \cite{hill2011bayesian}, which is a commonly used pseudo-real dataset as ground truth causal effects are inaccessible in most real studies. Our goal is to demonstrate the effectiveness of ACEE in estimating causal effects.

\paragraph*{Dataset.}
The IHDP dataset originates from a randomized controlled trial designed to assess the impact of home visits by specialist doctors on the cognitive test scores for premature infants. The experiment proposed by \cite{hill2011bayesian} involves simulating outcomes and introducing an artificial imbalance between treated and control groups by removing a subset of treated subjects. The dataset comprises 747 subjects (139 treated, 608 control), each described by 25 covariates representing various characteristics of the child and mother. 


\paragraph{Results.}
We select the following six features as covariates: \textit{Birth Weight, Head Circumference, Weeks Born Preterm, Birth Order, Neonatal Health Index, age of mom}. To take account of unmeasured confounding, we choose  $q =1$ for proxy in ACEE implementation. We report the absolute errors of estimated ATEs over 50 trials in Figure \ref{fig:ihdp_confounding}, from which we observe that ACEE helps improve ATE estimation accuracy via transfer learning and the results are stable.

\begin{figure}[H]
    \centering
    \includegraphics[width=0.6\textwidth]{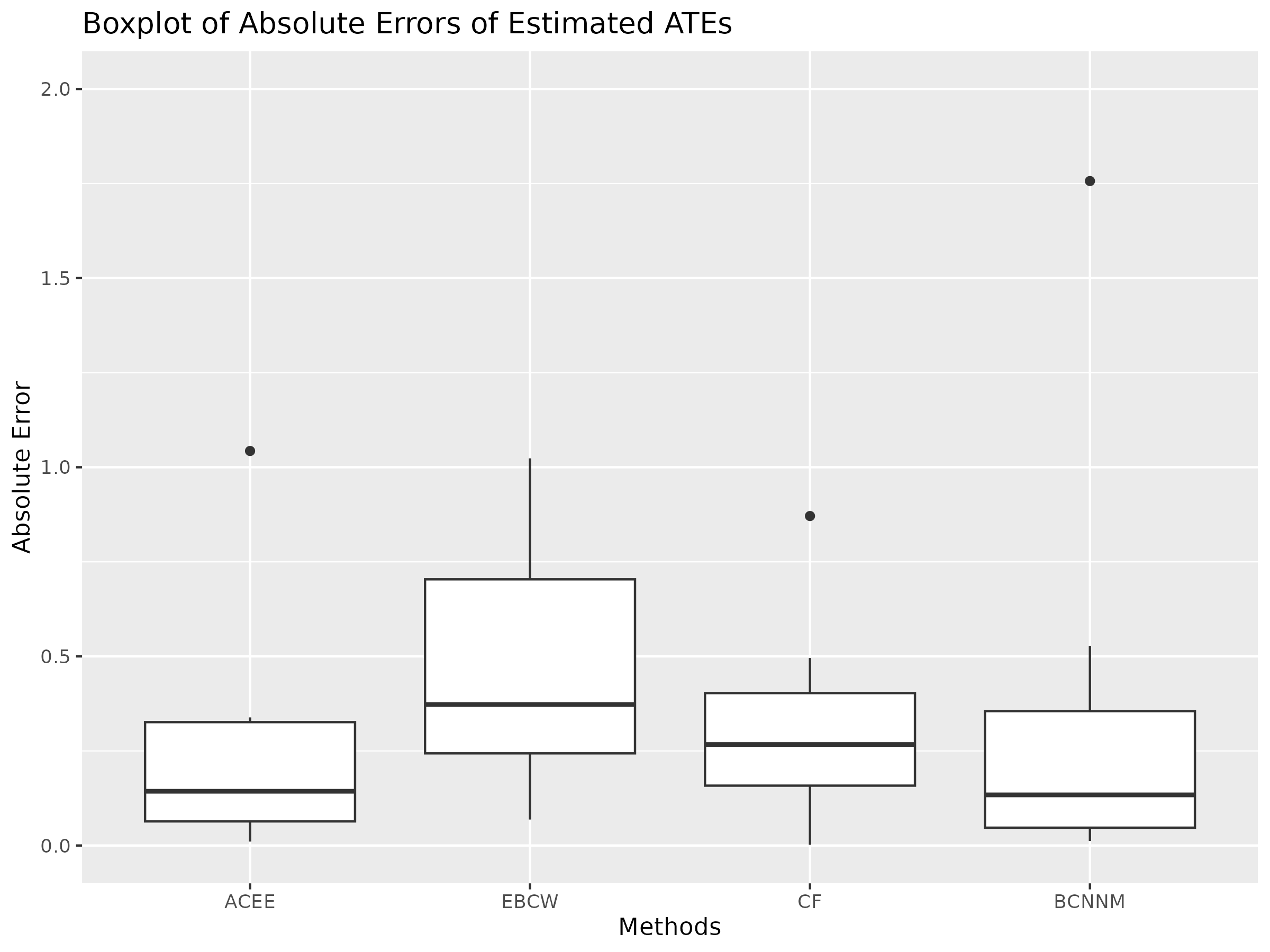}
    \caption{Box-plot of absolute errors of estimated ATEs over 50 trials using IHDP dataset.}
    \label{fig:ihdp_confounding}
\end{figure}

\section{Discussion}
This article presents an innovative method to improve the accuracy of causal effect estimation through knowledge transfer. By using synthetic data generated by a generative model trained on auxiliary data, our approach leverages additional information to enhance estimation precision. Our method assumes a known causal order, highlighting a future research direction: developing a causal discovery technique that employs knowledge transfer to determine causal order. Moreover, incorporating synthetic data to infer relationships and construct valid confidence intervals for estimated effects remains an uncharted and challenging endeavor.

\acks{This work was supported in part by NSF grant DMS-1952539 and NIH grants
R01AG069895, R01AG065636, R01AG074858, U01AG073079.}

\appendix

\section{Conditional Diffusion Models}
\label{sec:cdm}

Conditional diffusion models generate synthetic data by learning to reverse a stochastic diffusion process conditioned on observed data. Here, we briefly introduce the two diffusion processes. 

\paragraph{Forward Diffusion Process.} The forward diffusion process gradually adds noise to the variable $X_j$, eventually transforming it into white noise. This process is formalized as:
	\begin{align*}
		\mathrm{d} X_j(\tau)=-\frac{1}{2} X_j(\tau) \mathrm{d} \tau+\mathrm{d} W(\tau) \quad \text { with } \quad X_j(0) \sim \mathbb{P}(\cdot \mid \bX_{k^+},\bS)
	\end{align*}
where $W(\tau)$ is a Wiener process. The initial conditional distribution is denoted as $\mathbb{P}(\cdot \mid \bX_{k^+},\bS)$, and the marginal conditional distribution at a given time $\tau$ is denoted as $\mathbb{P}_{\tau}(\cdot \mid \bX_{k^+},\bS)$. Practically, the forward process terminates at a sufficiently large time $\bar{\tau}$.

\paragraph{Reverse Diffusion (Denoising) Process.} The reverse diffusion process, defined by $\tilde{X}_j(\tau) = X_j(\bar{\tau}-\tau)$, learns to reconstruct (denoise) the original distribution $\mathbb{P}(\cdot \mid \bX_{k^+},\bS)$ by reversing the forward diffusion:
	\begin{align*}
		\mathrm{d} \tilde{X}_j(\tau) =\left[\frac{1}{2}\tilde{X}_j(\tau) +\nabla \log p_{\bar{\tau}-\tau}\left(\tilde{X}_j(\tau) \mid \bX_{k^+},\bS \right)\right] \mathrm{d} t+\mathrm{d} \bar{W}(\tau) \quad \text { with } \quad \tilde{X}_j(0) \sim P_{\bar{\tau}}(\cdot \mid \bX_{k^+},\bS)
	\end{align*}
where $\bar{W}(\tau)$ is a time-reversed Wiener process, and $\nabla \log p_{\bar{\tau}-\tau}\left(\tilde{X}_j(\tau) \mid \bX_{k^+},\bS \right)$ is the conditional score function.

\paragraph{Sampling Step.} In practice, the conditional score function is unknown and is approximated by an estimator $\widehat{\theta}(X_j, \widehat{h}(\bX_{k^+},\bS), \tau)$. The synthetic data sampling process can thus be written as:
	\begin{align*}
		\mathrm{d} \tilde{X}_j(\tau) =\left[\frac{1}{2}\tilde{X}_j(\tau) +\widehat{\theta}(X_j, \widehat{h}(\bX_{k^+},\bS), \tau) \right] \mathrm{d} t+\mathrm{d} \bar{W}(\tau) \quad \text { with } \quad \tilde{X}_j(0) \sim P_{\bar{\tau}}(\cdot \mid \bX_{k^+},\bS)
	\end{align*}
	
	\paragraph{Estimation of conditional score function.} To estimate $\widehat{\theta}(X_j, \widehat{h}(\bX_{k^+},\bS), \tau)$, we minimize the following loss,
	\begin{align}
		\label{eq:loss_score_match}
		\int_{\underline{\tau}}^{\bar{\tau}} \mathrm{E}_{(X_j(0),\bX_{k^+},\bS)} [\mathrm{E}_{X_j(\tau) \mid X_j(0)}\left\|\nabla \log p_{X_j(\tau) \mid X_j(0)}(X_j(\tau) \mid X_j(0))-\theta(X_j(\tau),h(\bX_{k^+},\bS),\tau)\right\|^2] \mathrm{~d} \tau
	\end{align}

In practice, \eqref{eq:loss_score_match} is implemented using i.i.d. data points. For source data set, we denote the loss function 
\begin{align}
	\label{eq:loss_score_match_implement_source}
	\ell_s (X_j^s,\bX^s_{k^+},\bS^s;\theta,h) = \int_{\underline{\tau}}^{\bar{\tau}} \mathrm{E}_{X_j^s(\tau) \mid X_j^s(0)}\left\|\nabla \log p_{X_j^s(\tau) \mid X_j^s(0)}(X_j^s(\tau) \mid X_j^s(0))-\theta(X^s_j(\tau),h(\bX^s_{k^+},\bS^s),\tau)\right\|^2 \mathrm{~d} \tau
\end{align}
And for the target data set, we denote the loss function
\begin{align}
	\label{eq:loss_score_match_implement_target}
	\ell (X_j,\bX_{k^+},\bS;\theta,h) = \int_{\underline{\tau}}^{\bar{\tau}} \mathrm{E}_{X_j(\tau) \mid X_j(0)}\left\|\nabla \log p_{X_j(\tau) \mid X_j(0)}(X_j(\tau) \mid X_j(0))-\theta(X_j(\tau),h(\bX_{k^+},\bS),\tau)\right\|^2 \mathrm{~d} \tau
\end{align}

\paragraph{ReLU network architecture.} Consider the following class of ReLU neural networks, denoted by $\mathcal{F}:$
\begin{equation*}
	\begin{aligned}
		\mathcal{F}\left(M_t, W, \kappa, L, K\right):=\{ & \theta(z, h(\mathbf{x},\mathbf{s}), t)=\left(A_L \sigma(\cdot)+\mathbf{b}_L\right) \circ \cdots \circ\left(A_1\left[z, h(\mathbf{x},\mathbf{s})^{\top}, t\right]^{\top}+\mathbf{b}_1\right): \\
		& A_i \in \mathbb{R}^{d_i \times d_{i+1}}, \mathbf{b}_i \in \mathbb{R}^{d_{i+1}}, \max d_i \leq W, \sup _{z,\mathbf{x}, \mathbf{s}}\|\theta(z, h(\mathbf{x},\mathbf{s}), t)\|_{\infty} \leq M_t, \\
		& \left.\max _i\left\|A_i\right\|_{\infty} \vee\left\|\mathbf{b}_i\right\|_{\infty} \leq \kappa, \sum_{i=1}^L\left(\left\|A_i\right\|_0+\left\|\mathbf{b}_i\right\|_0\right) \leq K\right\},
	\end{aligned}
\end{equation*}
where $\sigma(\cdot)$ is the ReLU activation, $\|\cdot\|_{\infty}$ is the maximal magnitude of entries and $\|\cdot\|_0$ is the number of nonzero entries.  The complexity of this network class is controlled by the number of layers, the number of neurons of each layer, the magnitude of the network parameters, the number of nonzero parameters, and the magnitude of the neural network output.

\section{Justification of Structural Equation Models in \eqref{eq:sem}} 
\label{sec:canonical_DAG}
Under the presence of unobserved confounders, there may not exist a DAG over the observed variables 
that preserves all the conditional and marginal independencies in the marginal distribution of the full DAG 
over the observed and unobserved variables. 
To address this limitation, we assume that the joint distribution over the complete set of variables 
\( (\boldsymbol{X},\bar{\boldsymbol{H}}) \) is Markov to a DAG \( \bar{\mathcal{G}} \), 
where \( \bar{\boldsymbol{H}} \) represents a vector of \( \bar{K} \in \mathbb{N} \) unobserved variables.

Without further assumptions, estimating the DAG \( \bar{\mathcal{G}} \) from observational data alone 
is not possible, as \( \bar{\boldsymbol{H}} \) is unobserved. 
We begin by introducing terminologies related to the canonical exogenous DAG as introduced in \cite{agrawal2023decamfounder}.

\begin{definition}
	\( X_i \) has a completely hidden path to \( X_j \) in \( \bar{\mathcal{G}} \) if there exists a path 
	\[
	X_i \rightarrow \bar{H}_{k_1} \rightarrow \cdots \rightarrow \bar{H}_{k_m} \rightarrow X_j
	\]
	in \( \bar{\mathcal{G}} \), where \( 1 \leq i, j \leq p \), and \( 1 \leq k_1, \ldots, k_m \leq \bar{K} \).
\end{definition}

\begin{definition}
	Let \( \bX_C = \{X_c : c \in C\} \), where \( C \subseteq \{1,\ldots,p\} \). 
	\( \bX_C \) shares a hidden common cause in \( \bar{\mathcal{G}} \) if there exists an unobserved node \( \bar{H}_j \), 
	\( 1 \leq j \leq \bar{K} \), such that for each \( c \in C \), there is a directed path from \( \bar{H}_j \) to \( X_c \) 
	with all intermediate vertices in \( \{\bar{H}_k\}_{k=1}^{\bar{K}} \). We denote the collection of all maximal sets \( C \) such that \( \bX_C \) shares a hidden common cause 
	in \( \bar{\mathcal{G}} \) by \( \{C_1, \ldots, C_K\} \) for \( K \in \mathbb{N} \).
\end{definition}

Next, we introduce the canonical exogenous DAG, which is equivalent to the canonical DAG associated 
with the latent projection of \( \bar{\mathcal{G}} \) onto the observed nodes \citep{evans2016graphs}.

\begin{definition}[Canonical Exogenous DAG]
	The canonical exogenous DAG of \( \bar{\mathcal{G}} \) is a DAG \( \mathcal{G}' \) with:
	\begin{enumerate}
		\item an edge \( X_i \rightarrow X_j \) if \( X_i \) has a completely hidden path to \( X_j \) in \( \bar{\mathcal{G}} \), \( 1 \leq i, j \leq p \);
		\item a node \( H_k \) and edges \( H_k \rightarrow X_c \) for all \( X_c \in \bX_{C_k} \), \( 1 \leq k \leq K \).
	\end{enumerate}
\end{definition}

We denote \( \{H_i\}_{i=1}^{K} \) as \( \bH \), and thus we can transform \( \bar{\mathcal{G}} \) 
into another DAG \( \mathcal{G}' \) that satisfies the following three properties:
\begin{enumerate}
	\item all unobserved variables are sources in \( \mathcal{G}' \),
	\item the partial ordering of observed variables in \( \mathcal{G}' \) matches the partial ordering 
	in \( \bar{\mathcal{G}} \),
	\item there exists a joint distribution \( \mathbb{Q}(\bX,\bH) \), 
	Markov with respect to \( \mathcal{G}' \), such that the marginal distribution 
	\( \mathbb{Q}(\bX) \) is equal to \( \mathbb{P}(\bX) \) almost everywhere.
\end{enumerate}

\begin{figure}[H]
	\centering
	\begin{subfigure}
		\centering
		\includegraphics[width=0.49\textwidth]{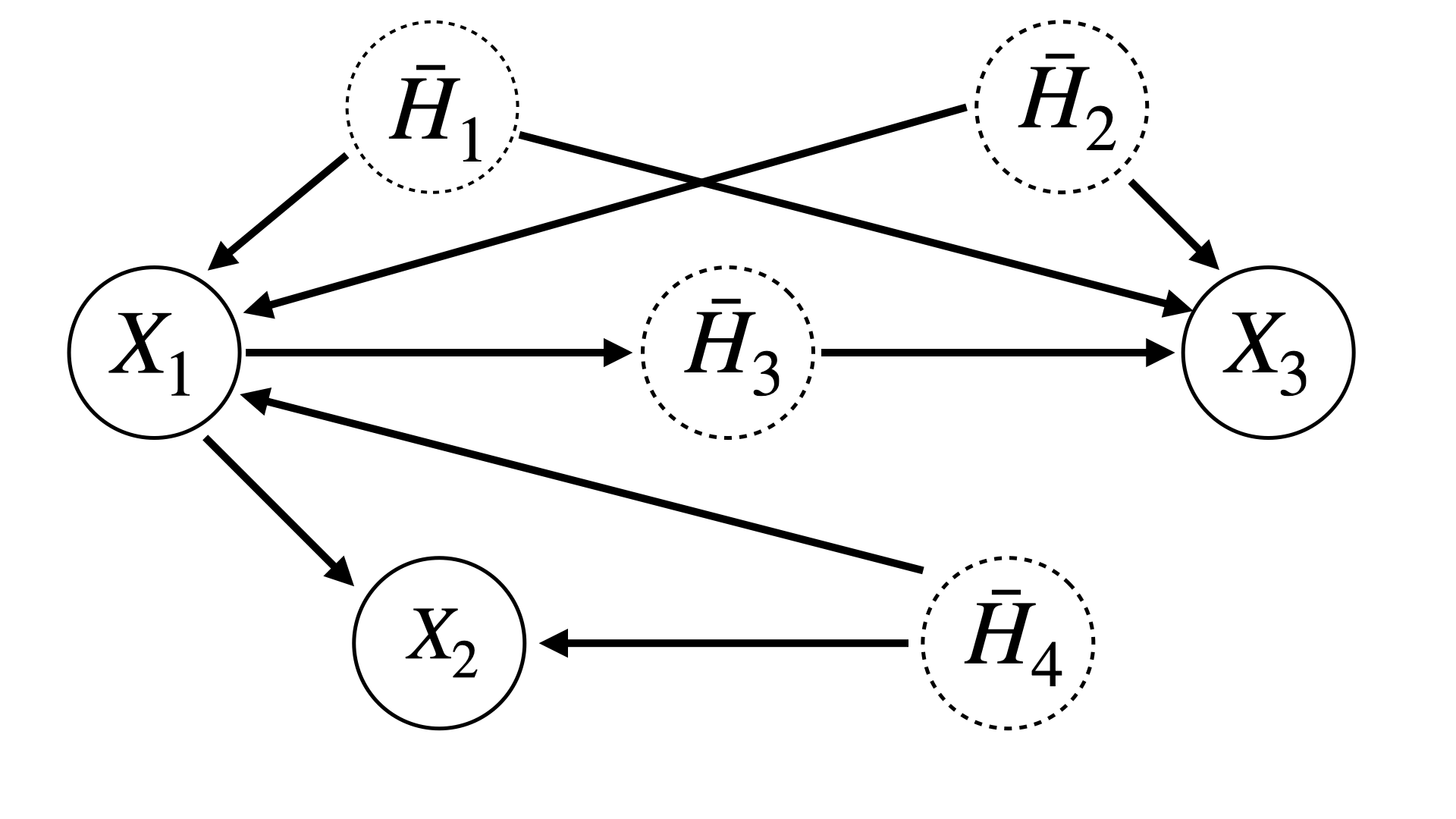}
	\end{subfigure}
	\hfill
	\begin{subfigure}
		\centering
		\includegraphics[width=0.49\textwidth]{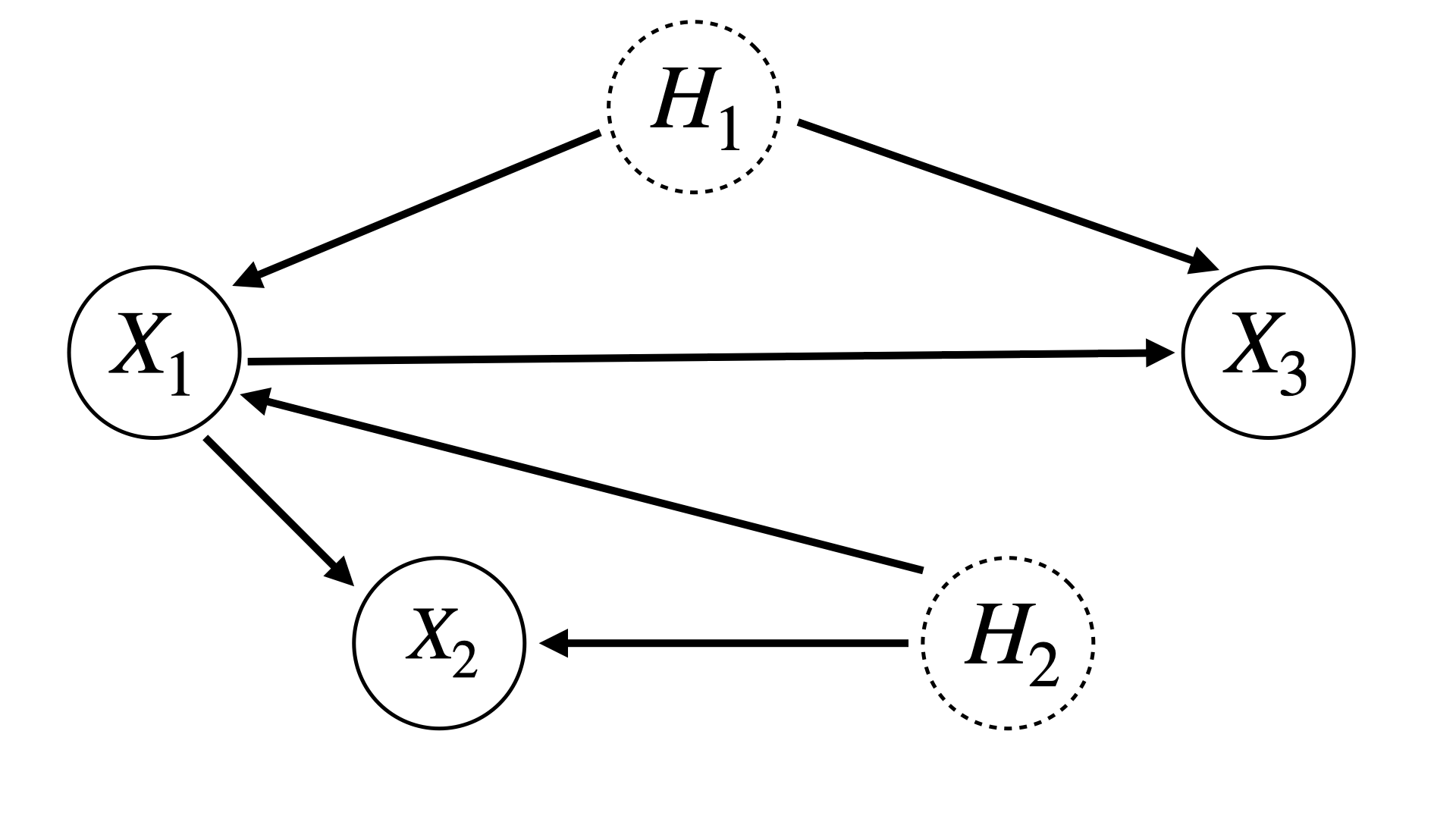}
	\end{subfigure}
	\caption{An illustration of transforming a DAG to its canonical exogenous DAG. Left panel: original DAG; Right panel: corresponding canonical exogenous DAG.}
	\label{fig:canonical_dag}
\end{figure}

We treat \( \mathbb{Q}(\bX,\bH) \) as the canonical representation of the joint distribution over observed variables and unmeasured confounders, since we cannot distinguish between the distributions \( \mathbb{P}(\bX,\bar{\bH}) \) and \( \mathbb{Q}(\bX,\bH) \) from observational data. The discussion of the canonical representation justifies our definition of structural equation models in \eqref{eq:sem}.

\section{Discussion on Causal Effect Definition in \eqref{eq:DAG_ate_conditionalH}}
\label{sec:ce_def}
In this subsection, we provide a discussion on definition of causal effect in \eqref{eq:DAG_ate_conditionalH} and relates it to the concept of total causal effect in the causal literature\citep{pearl2009causality}. Additionally, we introduce indirect causal effect in the causal literature\citep{pearl2009causality}. 

\begin{assumption}[Strong Ignorability in DAG]
\label{ass:strong_ignorability_DAG}
    Given a causal order $\pi$ and joint distribution of $\bX$, $\mathbb{P}_{\bX}$ consistent with the true DAG $\mathcal{G} = (\mathcal{V},\mathcal{E})$, suppose that for all $k,j$ with $X_k$ preceding $X_j$ in $\pi$,
    \begin{equation}
        \begin{split}
            (k \rightarrow j) \not \in \mathcal{E} \Leftrightarrow X_{k} \perp_{\mathbb{P}_{\bX}} X_{j} | \bX_{\pa(j)/k}
        \end{split}
    \end{equation}
    where $\bX_{\pa(j)/k}$ represents the parent set of $X_j$ excluding $X_k$.
\end{assumption} 
The condition $(k \rightarrow j) \not \in \mathcal{E} \Rightarrow X_{k} \perp_{\mathbb{P}_{\bX}} X_{j} | \bX_{\pa(j)/k}$ is called the ``local'' Markov condition and often taken as the definition of Bayesian networks\citep{howard2005influence}, while $(k \rightarrow j) \not \in \mathcal{E} \Leftarrow X_{k} \perp_{\mathbb{P}_{\bX}} X_{j} | \bX_{\pa(j)/k}$ ensures that $\bX_{\pa(j)}$ is the minimal set of predecessors of $X_{j}$ that renders $X_{j}$ independent of all its other predecessors. We say that $\mathbb{P}(\bX)$ is Markov with respect to $\mathcal{G}$ if Assumption \ref{ass:strong_ignorability_DAG} holds and establishes the identification of causal identification under this assumption.

\begin{definition}[blocking]
    A set $\mathbb{S}$ of nodes is said to block a path if either
    \begin{enumerate}
        \item The path contains at least one arrow-emitting node that is in $\mathbb{S}$, or
        \item The path contains at least one collider that is outside $\mathbb{S}$ and has no descendant in $\mathbb{S}$.
    \end{enumerate}
\end{definition}

To illustrate, the path $X_1 \rightarrow X_2 \rightarrow X_3 \rightarrow X_4$ is blocked by $\mathbb{S} = \{X_2\}$ and $\mathbb{S}=\{X_3\}$ since each variable emits an arrow along the path. Thus we can infer the conditional independencies that $X_1 \perp X_4 | X_2$ and $X_1 \perp X_4 | X_3$ in any probability function that can be generated from this DAG model. Likewise, the path $X_1 \rightarrow X_2 \leftarrow X_3 \rightarrow X_4$ is blocked by the null set $\{ \emptyset \}$ but is not blocked by $\mathbb{S} = \{X_2\}$ since $X_2$ is a collider. Thus the marginal dependence $X_1 \perp X_4$ holds in the distribution while $X_1 \perp X_4 | X_2$ may not. The handling of colliders reflects Berkson's paradox\citep{berkson1946limitations}, whereby observations on the common consequence of two independent causes render those causes dependent.

\begin{definition}[Admissible Sets]
\label{def:admissible}
    A set $\mathbb{S}$ is admissible with respect to $(X_{k},X_j)$ where $k \in \nd(j)$ if two conditions hold:
    \begin{enumerate}
        \item No element of $\mathbb{S}$ is a descendant of $X_{k}$.
        \item The elements of $\mathbb{S}$ block all back-door paths from $X_{k}$ to $X_{j}$
    \end{enumerate}
\end{definition}

The back-door paths refer to the connections in the diagram between $X_{k}$ and $X_{j}$ that contains an arrow into $X_{k}$, which could carry spurious associations from $X_{k}$ to $X_{j}$. Blocking the back-door paths ensures that the measured association between $X_{k}$ and $X_{j}$ is purely causal, which is the target quantity carried by the paths directed along the arrows.

 Denote the structural equation model as $\mathcal{M}$ and we define the interventions and counterfactuals through a mathematical operator called $\operatorname{do}(X_k=x)$, which simulates physical interventions by deleting certain functions from the model, replacing them with a constant $X_k=x$, while keeping the rest of the model unchanged. The resulting model is denoted $\mathcal{M}_{x}$ and then the post-intervention distribution resulting from the action $\operatorname{do}(X_k =x)$ is given by the equation 
 \begin{equation*}
     P_{\mathcal{M}}(X_j|\operatorname{do}(X_k=x)) = P_{\mathcal{M}_x}(X_j)
 \end{equation*}

The back-door criterion provides a simple solution to identification problems and is summarized in the following theorem.

\begin{theorem}
\label{thm:identification_admissible_set}
    Under Assumption \ref{ass:strong_ignorability_DAG}, the causal effect of $X_{k}$ on $X_{j}$ is identified by
    \begin{equation*}
        P(X_{j} \leq t | \operatorname{do}(X_{k} = x)) = \E_{\mathbb{S}} [P(X_{j} \leq t|X_{k}=x, \mathbb{S})]
    \end{equation*}
    where $\mathbb{S}$ is an admissible set with respect to $(X_{k},X_{j})$.
\end{theorem}

\begin{proof}[Proof of Theorem \ref{thm:identification_admissible_set}]
	Consider a Markovian model $\mathcal{G}'$ in which $\mathbb{T}$ stands for the set of parents of $X_{k}$. Then by adjustment for direct causes (Theorem 3.2.2 in \cite{pearl2009causality}), we have
	\begin{equation*}
		P(X_{j} \leq t | \operatorname{do}(X_{k} = x)) = \E_{\mathbb{T}} [P(X_{j} \leq t|X_{k}=x, \mathbb{T})]
	\end{equation*}
	Since any path from $X_{j}$ to $\mathbb{T}$ in $\mathcal{G}'$ that is unblocked by $\{X_{k},\mathbb{S}\}$ can be extended to a back-door path from $X_{j}$ to $X_{k}$ unblocked by $\mathbb{S}$, we have $X_{j} \perp \mathbb{T}|X_{k},\mathbb{S}$ in $\mathcal{G}'$. Thus
	\begin{equation*}
		\begin{split}
			P(X_{j} \leq t | \operatorname{do}(X_{k} = x)) & = \E_{\mathbb{T}} [\E[P(X_{j} \leq t|X_{k}=x,\mathbb{S}) |X_{k}=x,\mathbb{T}]] \\
			& = \E_{\mathbb{S}} [P(X_{j} \leq t|X_{k}=x, \mathbb{S})]
		\end{split}
	\end{equation*}
	where the second equality follows from the fact that $X_{k} \perp \mathbb{S}|\mathbb{T}$ in $\mathcal{G}'$ since $\mathbb{S}$ consists of nondescendants of $X_{k}$.
\end{proof}

The construction of $\mathcal{G}'$ ensures the fulfillment of Assumption \ref{ass:strong_ignorability_DAG} and we have the following result. 

\begin{lemma}
\label{lem:conditional_equivalence_DAG_confound}
    $(\bH,\bX_{k^-})$ is an admissible set with respect to $X_{k}$ and $X_{j}$ in $\mathcal{G}'$.
\end{lemma}

\begin{proof}[Proof of Lemma \ref{lem:conditional_equivalence_DAG_confound}]
	First of all, note that no element of $(\bH,\bX_{k^-})$ is a descendant of $X_{k}$. Now we just need to show that $(\bH,\bX_{k^-})$ block all back-door paths from $X_{k}$ to $X_{j}$. This is again obvious since $\pa_{\mathcal{G}'}(X_{k}) \subseteq (\bH,\bX_{k^-})$.
\end{proof}

With Theorem \ref{thm:identification_admissible_set}, Lemma \ref{lem:conditional_equivalence_DAG_confound}, then the total effect from $X_{k}$ to $X_{j}$ can be expressed as follows
\begin{equation}
    \label{eq:TE2}
    \begin{split}
        \tau(k,j;x_1,x_0) & = \E[X_{j} | \operatorname{do}(X_{k}=x_1)] - \E[X_{j}|\operatorname{do}(X_{k}=x_0)] \\
        & = \E[\E[X_{j}|\bH,\bX_{k^-},X_{k} =x_1] - \E[X_{j}|\bH,\bX_{k^-},X_{k} =x_0]]
    \end{split}
\end{equation}

which is exactly the same as our definition in \eqref{eq:DAG_ate_conditionalH}.

Under Assumption \ref{ass:strong_ignorability_DAG}, we define the indirect effect to be the expected difference of $X_{j}$ when fixing $X_{k}$ at its reference level $x_0$ and change $\bX_{j^-/k^+}$ to the value that $\bX_{j^-/k^+}$ would attain under $X_{k} =x_1$
\begin{equation}
	\label{eq:IE}
	\begin{split}
		\tau_{\operatorname{i}}(k,j;x_1,x_0) & = \E_{\bX_{k^-},\bH}[\E_{\bX_{j^-/k^+}}[\E[X_{j}|X_{k}=x_0,\bX_{j^-/k^+},\bX_{k^-}, \bH]|X_{k}=x_1,\bX_{k^-},\bH]] \\ 
		& - \E_{\bX_{k^-},\bH}[\E_{\bX_{j^-/k^+}}[\E[X_{j}|X_{k}=x_0,\bX_{j^-/k^+},\bX_{k^-},\bH]|X_{k}=x_0,\bX_{k^-},\bH]].
	\end{split}
\end{equation}

Comparing \eqref{eq:DE} and \eqref{eq:IE}, one can obtain the following proposition.
\begin{proposition}
	\label{prop:TE_DE_IE}
	Given a causal order $\pi$ and joint distribution $\mathbb{P}_X$ consistent with the true DAG $\mathcal{G} = (\mathcal{V},\mathcal{E})$, under Assumption \ref{ass:strong_ignorability_DAG}, we have
	\begin{equation}
		\begin{split}
			\tau(k,j;x_1,x_0) = \tau_{\operatorname{d}}(k,j;x_1,x_0) - \tau_{\operatorname{i}}(k,j;x_0,x_1) \\
			\tau(k,j;x_1,x_0) = \tau_{\operatorname{i}}(k,j;x_1,x_0) - \tau_{\operatorname{d}}(k,j;x_0,x_1)
		\end{split}
	\end{equation}
\end{proposition}
\begin{proof}[Proof of Proposition \ref{prop:TE_DE_IE}]
	Notice that 
	\begin{equation*}
		\begin{split}
			\tau_{\operatorname{d}}(k,j;x_1,x_0) & = \E_{\bX_{k^-},\bH}[\E_{\bX_{j^-/k^+}}[\E[X_{j}|X_{k}=x_1,\bX_{j^-/k^+},\bX_{k^-},\bH]|X_{k}=x_0,\bX_{k^-},\bH]] \\ 
        & - \E_{\bX_{k^-},\bH}[\E_{\bX_{j^-/k^+}}[\E[X_{j}|X_{k}=x_0,\bX_{j^-/k^+},\bX_{k^-},\bH]|X_{k}=x_0,\bX_{k^-},\bH]]. \\
			\tau_{\operatorname{i}}(k,j;x_1,x_0) & = \E_{\bX_{k^-},\bH}[\E_{\bX_{j^-/k^+}}[\E[X_{j}|X_{k}=x_0,\bX_{j^-/k^+},\bX_{k^-}, \bH]|X_{k}=x_1,\bX_{k^-},\bH]] \\ 
		& - \E_{\bX_{k^-},\bH}[\E_{\bX_{j^-/k^+}}[\E[X_{j}|X_{k}=x_0,\bX_{j^-/k^+},\bX_{k^-},\bH]|X_{k}=x_0,\bX_{k^-},\bH]].
		\end{split}
	\end{equation*}
	Then we can obtain
	\begin{equation*}
		\begin{split}
			\tau(k,j;x_1,x_0) 
			& = \E[\E[X_{j}|\bH,\bX_{k^-},X_{k} =x_1] - \E[X_{j}|\bH,\bX_{k^-},X_{k} =x_0]] \\
			& = \E_{\bX_{k^-},\bH}[\E_{\bX_{j^-/k^+}}[\E[X_{j}|X_{k}=x_1,\bX_{j^-/k^+},\bX_{k^-},\bH]|X_{k}=x_1,\bX_{k^-},\bH]] \\
			& - \E_{\bX_{k^-},\bH}[\E_{\bX_{j^-/k^+}}[\E[X_{j}|X_{k}=x_0,\bX_{j^-/k^+},\bX_{k^-},\bH]|X_{k}=x_0,\bX_{k^-},\bH]] \\
			& = \tau_{\operatorname{d}}(k,j;x_1,x_0) - \tau_{\operatorname{i}}(k,j;x_0,x_1)
		\end{split}
	\end{equation*}
	One can prove $\tau(k,j;x_1,x_0) = \tau_{\operatorname{i}}(k,j;x_1,x_0) - \tau_{\operatorname{d}}(k,j;x_0,x_1)$ in the same fashion.
\end{proof}

\section{Technical Proofs}
\label{sec:technical_proofs}
\begin{proof}[Proof of Lemma \ref{lem:proxy}]
	Since \(S_Y = E[Y\mid\bH]\) depends only on \(\bH\), the tower property yields 
	\[
	E[S_Y\mid\bX,\bH,D]=S_Y \quad \text{and} \quad E[Y\mid\bX,\bH,D] = E[Y-S_Y\mid\bX,\bH,D] + S_Y.
	\]
	Thus,
	\begin{eqnarray*}
		\tau & = & E\Bigl[E[Y\mid\bX,\bH,D=1]-E[Y\mid\bX,\bH,D=0]\Bigr] \\
		&= & E\Bigl[E[Y-S_Y\mid\bX,\bH,D=1]-E[Y-S_Y\mid\bX,\bH,D=0]\Bigr].
	\end{eqnarray*}
	Since \(\bS_{-Y}\) is a function of \(\bH\), we have 
	\[
	E[Y-S_Y\mid\bX,\bH,D] = E[Y-S_Y\mid\bX,\bH,\bS_{-Y},D].
	\]
	By Assumption \ref{assu:proxy},
	\[
	E[Y-S_Y\mid\bX,\bH,\bS_{-Y},D] = E[Y-S_Y\mid\bX,\bS_{-Y},D].
	\]
	Noting that \(\bS_Y\) is included in \(\bS = (\bS_{-Y}^\top, \bS_Y)^\top\), it follows that
	\[
	E[Y-S_Y\mid\bX,\bS_{-Y},D] + \bS_Y = E[Y\mid\bX,\bS,D].
	\]
	Therefore,
	\[
	\tau = E\Bigl[E[Y\mid\bX,\bS,D=1]-E[Y\mid\bX,\bS,D=0]\Bigr],
	\]
	which completes the proof.
\end{proof}

\begin{proof}[Proof of Example \ref{example:proxy}]
	Firstly, we prove that there exists an function $\widetilde{f}_j$ such that
	\begin{equation}
		\label{eq:covariate_substitute}
		X_j = \epsilon_j + g_j(\bH) + \widetilde{f}_j(\epsilon_1,g_1(\bH),\ldots,\epsilon_{j-1},g_{j-1}(\bH))
	\end{equation}
    The proof follows by inducing on the number of covariates. For $p=1$, $X_1 = g_1(\bH)+\epsilon_1$, and Eq. \eqref{eq:covariate_substitute} trivially holds. For $p = 2$,
    \begin{equation*}
    	\begin{split}
    		X_2 & = g_2(\bH) + \epsilon_2 + f_2(X_1) \\
    		& = g_2(\bH) + \epsilon_2 + f_2(g_1(\bH)+\epsilon_1)
    	\end{split}
    \end{equation*}
    The claim holds by setting $\widetilde{f}_2(\epsilon_1,g_1(\bH)) = f_2(g_1(\bH)+\epsilon_1)$. Suppose that Eq. \eqref{eq:covariate_substitute} holds for $p-1$ covariates. Then we just need to prove that it holds for $p$ covariates. For this, note that
    \begin{equation*}
    	\begin{split}
    		X_p & = g_p(\bH) + \epsilon_p + f_p(X_1,\ldots,X_{p-1})\\
    		&  = g_p(\bH) + \epsilon_p + f_p(g_1(\bH)+\epsilon_1,\ldots,\epsilon_{p-1} + g_{p-1}(\bH) + \widetilde{f}_{p-1}(\epsilon_1,g_1(\bH),\ldots,\epsilon_{p-2},g_{p-2}(\bH)))
    	\end{split}
    \end{equation*}
    where the last equality follows from the inductive hypothesis. Thus by setting
    \begin{equation*}
    	\widetilde{f}_p (\epsilon_1,g_1(\bH),\ldots,\epsilon_{p-1},g_{p-1}(\bH)) = f_p(g_1(\bH)+\epsilon_1,\ldots,\epsilon_{p-1} + g_{p-1}(\bH) + \widetilde{f}_{p-1}(\epsilon_1,g_1(\bH),\ldots,\epsilon_{p-2},g_{p-2}(\bH)))
    \end{equation*}
    the claim follows. Similarly, we can prove that there exists $\widetilde{f}_{p+1}$ that
    \begin{equation*}
    	\begin{split}
    		P(D=1) & = \frac{1}{1+ \exp(\widetilde{f}_{p+1}(\epsilon_1,g_1(\bH),\ldots,\epsilon_{p},g_{p}(\bH)))}\frac{1}{1+\exp( g_{p+1}(\bH))}
    	\end{split}
    \end{equation*}
	Thus we have that $\E[X_j|\bH] = \E[X_j \mid g_1(\bH),\ldots,g_j(\bH)]$ for $j=1,\ldots,p$.
    Then we prove that there exists an $r_j$ such that $g_j(\bH) = r_j(S_1,\ldots,S_j)$. We still prove this claim by inducting on the number of covariates $p$. For $p=1$, $g_1(\bH)=S_1$. For $p=2$, 
    \begin{equation*}
    	\begin{split}
    		S_2 & = \E[X_2 \mid \bH]  \\
    		& = \E[X_2 \mid g_1(\bH),g_2(\bH)] \\
    		& = g_2(\bH) + \E[f_2(X_1) \mid g_1(\bH)] \\
    		& = g_2(\bH) + \E[f_2(X_1) \mid S_1]
    	\end{split}
    \end{equation*}
    The claim holds by setting $r_2(S_1,S_2) = S_2 - \E[f_2(X_1) \mid S_1]$. Suppose the claim holds for all with $p-1$ covariates. Then it suffices to show that there exists $r_p$ such that $g_p(\bH) = r_p(S_1,\ldots,S_p)$. By the inductive hypothesis, there exists $\{r_j\}_{j=1}^p$ such that
    \begin{equation*}
    	g_{j}(\bH) = r_j(S_1,\ldots,S_{j}), \forall j = 1,\ldots,p-1
    \end{equation*}
    Then we have
    \begin{equation*}
    	\begin{split}
    		S_p & = \E[X_p \mid \bH] \\
    		& = \E[X_p \mid g_1(\bH),\ldots,g_p(\bH)] \\
    		& = g_p(\bH) + \E[f_p(X_1,\ldots,X_{p-1}) \mid g_1(\bH),\ldots,g_{p-1}(\bH)]\\
    		& = g_p(\bH) + \E[f_p(X_1,\ldots,X_{p-1}) \mid \{r_j(S_1,\ldots,S_j)\}_{j=1}^{p-1}]
    	\end{split}
    \end{equation*}
    Then the claim follows by setting $r_p(S_1,\ldots,S_p) = S_p - \E[f_p(X_1,\ldots,X_{p-1}) \mid \{r_j(S_1,\ldots,S_j)\}_{j=1}^{p-1}]$. Similarly,
    \begin{equation*}
    	\begin{split}
    		S_{p+1} & = \E[D|\bH] \\
    		& = \E[D \mid g_1(\bH),\ldots,g_{p+1}(\bH)] \\
    		& = \frac{1}{1+\exp(g_{p+1}(\bH))} \E[ \frac{1}{1+\exp(f_{p+1}(X_1,\ldots,X_p))}\mid g_1(\bH),\ldots,g_p(\bH)]\\
    		& = \frac{1}{1+\exp(g_{p+1}(\bH))} \E[ \frac{1}{1+\exp(f_{p+1}(X_1,\ldots,X_p))}\mid \{r_j(S_1,\ldots,S_j)\}_{j=1}^{p}]
    	\end{split}
    \end{equation*}
    Thus we have $g_{p+1}(\bH) = r_{p+1}(S_1,\ldots,S_{p+1})$ by setting
    \begin{equation*}
    	r_{p+1}(S_1,\ldots,S_{p+1}) = \log(\frac{\E[ \frac{1}{1+\exp(f_{p+1}(X_1,\ldots,X_p))}\mid \{r_j(S_1,\ldots,S_j)\}_{j=1}^{p}]}{S_{p+1}}-1)
    \end{equation*} 
    Finally, we have 
    \begin{equation*}
    	\begin{split}
    		S_{p+2} & = \E[Y|\bH] \\
    		& = g(\bH)+ \E[f(X_1,\ldots,X_p,D)|g_1(\bH),\ldots,g_{p+1}(\bH)] \\
    			& = g(\bH)+ \E[f(X_1,\ldots,X_p,D)| \{r_j(S_1,\ldots,S_j)\}_{j=1}^{p+1}]
    	\end{split}
    \end{equation*}
    Thus $g(\bH) = r(S_1,\ldots,S_{p+2})$ by setting 
    \begin{equation*}
    	r(S_1,\ldots,S_{p+2}) = S_{p+2} - \E[f(X_1,\ldots,X_p,D)| \{r_j(S_1,\ldots,S_j)\}_{j=1}^{p+1}]
    \end{equation*}
    This leads to 
    \begin{equation*}
    	\begin{split}
    		Y - S_{p+2} = f(X_1,\ldots,X_p,D) - \E[f(X_1,\ldots,X_p,D)| \{r_j(S_1,\ldots,S_j)\}_{j=1}^{p+1}] + \varepsilon.
    	\end{split}
    \end{equation*}
    This concludes that  $Y - \E[Y \mid \bH]$ is conditionally independent of $\bH$ given $(\bX,\E[\bX \mid \bH], \E[D \mid \bH],D)$, which means that Assumption \ref{assu:proxy} holds.
\end{proof}

\begin{proof}[Proof of Lemma \ref{lem:proxy_DAG}]
	The proof following the same proof of Lemma \ref{lem:proxy}
\end{proof}

\begin{proof}[Proof of Proposition \ref{prop:de_edge}]
	Assume $(k \rightarrow j) \not \in \mathcal{E}$. By Lemma \ref{lem:ci_equivalence} and set $\mathcal{M} = j^-/k \cup \bH$,
	we have 
	\begin{equation*}
		\begin{split}
			\E[X_{j}|X_{k}=x_1,\bX_{j^-/k^+},\bX_{k^-},\bH] = \E[X_{j}|X_{k}=x_0,\bX_{j^-/k^+},\bX_{k^-},\bH]
		\end{split}
	\end{equation*}
	for all $x_0, x_1 \in \mathbb{R}$. Thus $\tau_{\operatorname{d}}(k,j;x_1,x_0) = 0$ for all $x_0, x_1 \in \mathbb{R}$.   
\end{proof}

\begin{proof}[Proof of Theorem \ref{thm:acee_rate}]
	Notice that
	\begin{equation*}
		\begin{split}
			& \E |\frac{1}{n}\sum_{i=1}^n \widehat{\eta}_M(\mathbf{X}_{k^-,i},\widehat{
			\mathbf{S}}_{\cdot,i}) - \E[\eta_0(\mathbf{X}_{k^-,i},
			\mathbf{S}_{\cdot,i})] | \\
			\leq & \E|\frac{1}{n}\sum_{i=1}^n \widehat{\eta}_M(\mathbf{X}_{k^-,i},\widehat{
				\mathbf{S}}_{\cdot,i})  - \frac{1}{n}\sum_{i=1}^n \widehat{\eta}(\mathbf{X}_{k^-,i},\widehat{
				\mathbf{S}}_{\cdot,i})|  \\
				& + \E|\frac{1}{n}\sum_{i=1}^n \widehat{\eta}(\mathbf{X}_{k^-,i},\widehat{
					\mathbf{S}}_{\cdot,i}) - \E_n[\widehat{\eta}(\mathbf{X}_{k^-,1},\mathbf{S}_{\cdot,1})] | \\
					& + \E|\E_n[\widehat{\eta}(\mathbf{X}_{k^-,1},\mathbf{S}_{\cdot,1})]  - \eta_0 (\bX,\bS)| 
		\end{split}
	\end{equation*}
	
	By proof of Proposition 5.4 in \cite{fu2024unveil}, $\widehat{\mathbb{P}}(\bX_{k^-},\bS)$ has subgaussian tails for any fixed $(\bX_{k^-},\bS)$ uniformly, so that for the first term, 
	\begin{align*}
		\E|\frac{1}{n}\sum_{i=1}^n \widehat{\eta}_M(\mathbf{X}_{k^-,i},\widehat{
			\mathbf{S}}_{\cdot,i})  - \frac{1}{n}\sum_{i=1}^n \widehat{\eta}(\mathbf{X}_{k^-,i},\widehat{
			\mathbf{S}}_{\cdot,i})|  = O(\frac{1}{\sqrt{M}})
	\end{align*}
	For the second term, by Assumption \ref{assu:proxy_estimation_DAG}, we have 
	\begin{align*}
		\E|\frac{1}{n}\sum_{i=1}^n \widehat{\eta}(\mathbf{X}_{k^-,i},\widehat{
			\mathbf{S}}_{\cdot,i}) - \E_n[\widehat{\eta}(\mathbf{X}_{k^-,1},\mathbf{S}_{\cdot,1})] |  = O(\frac{1}{\sqrt{n}})
	\end{align*}
	
	For the third term, following the same procedures of deriving the estimation error of the posterior mean  in the proof of Proposition 5.4 in \cite{fu2024unveil} and Lemma \ref{lem:diffusion_transfer}, we can have
	\begin{equation*}
		\E|\eta(\bX) - \eta_0 (\bX)| = \mathcal{T}(x_1,x_0,\bX_{k^-},\bS) O\left(n^{-\frac{\beta}{1+d_{h}+2 \beta}} (\log (n))^{\max(11,(\beta+5)/2)}+\varepsilon_s\right)
	\end{equation*}
	where $\mathcal{T}(x_1,x_0,\bX_{k^-},\bS) = \max \{\mathcal{T}(x_1,\bX_{k^-},\bS),\mathcal{T}(x_0,\bX_{k^-},\bS)\}$.
	To conclude, we have
	\begin{equation*}
		\E[|\widehat{\tau}-\tau|] = O\left(n^{-\frac{\beta}{1+d_{h}+2 \beta}} (\log (n))^{\max(11,(\beta+5)/2)}+\frac{1}{\sqrt{n}}+\frac{1}{\sqrt{M}} + \epsilon_s\right)
	\end{equation*}
	This completes the proof.
\end{proof}

\begin{proof}[Proof of Theorem \ref{thm:consistency_bad_generation}]
	For notation simplicity, let assumptions in Theorem \ref{thm:consistency_bad_generation} are satisfied for $(\bX_1,\bS_1)$. Let $r_n^1 = \max\{\|(\bX_j,\bS_j) - (\bX_1,\bS_1)\| \mid j \in \mathcal{J}_N^1(\bX_1,\widehat{\bS}_1) \}$, then $r_n^1 \xrightarrow{p} 0$ as $n \rightarrow \infty$. By the continuity of $\mu_1$ and $\bar{\mu}_1$ at $(\bX_1,\bS_1)$, for any $\epsilon > 0$, there exists $\delta > 0$ such that $|(\mu_1(\bX,\bS) - \bar{\mu}_{1}(\bX, \bS)) - (\mu_1(\bX_1,\bS_1) - \bar{\mu}_{1}(\bX_1, \bS_1))| \leq \epsilon$ if $\|(\bX,\bS) - (\bX_1,\bS_1)\| < \delta$. Let $R_n^1$ be the event that $r_n^1 < \delta$, then we have $P(R_n^1) \rightarrow 1$ as $n \rightarrow \infty$.
   Note that
   \begin{equation*}
       \begin{split}
           \left|\E [\frac{1}{N} \sum_{i \in \mathcal{J}_N^1(\bX_1, \widehat{\bS}_1)} \Bigl(Y_i - \bar{\mu}_{1}(\bX_i, \bS_i)\Bigr) \mid R_n^1] - \Bigl(\mu_1(\bX_1,\bS_1) - \bar{\mu}_{1}(\bX_1, \bS_1)\Bigr)\right| < \epsilon
       \end{split}
   \end{equation*}
   and
   \begin{equation*}
   	\begin{split}
   		\operatorname{Var} [\frac{1}{N} \sum_{i \in \mathcal{J}_N^1(\bX_1, \widehat{\bS}_1)} \Bigl(Y_i - \bar{\mu}_{1}(\bX_i, \bS_i)\Bigr) \mid R_n^1] = \frac{1}{N^2} \sum_{i \in \mathcal{J}_N^1(\bX_1, \widehat{\bS}_1)} \operatorname{Var} [Y_i - \bar{\mu}_{1}(\bX_i, \bS_i) \mid R_n^1] = o(1)
   	\end{split}
   \end{equation*}
   Since $\epsilon > 0$ can be arbitrarily small, this leads to 
   \begin{equation*}
   	\frac{1}{N} \sum_{i \in \mathcal{J}_N^1(\bX_1, \widehat{\bS}_1)} \Bigl(Y_i - \bar{\mu}_{1}(\bX_i, \bS_i)\Bigr) \xrightarrow{p} \mu_1(\bX_1,\bS_1) - \bar{\mu}_{1}(\bX_1, \bS_1)
   \end{equation*}
   By Assumption \ref{assu:uncorrected_estimator}, we have
   \begin{equation*}
       \begin{split}
       	\frac{1}{N} \sum_{i \in \mathcal{J}_N^1(\bX_1, \widehat{\bS}_1)} \Bigl(Y_i - \widehat{\mu}_{1}(\bX_i, \widehat{\bS}_i)\Bigr) & = 	\frac{1}{N} \sum_{i \in \mathcal{J}_N^1(\bX_1, \widehat{\bS}_1)} \Bigl(Y_i - \bar{\mu}_{1}(\bX_i, \bS_i)\Bigr)\\ & + \frac{1}{N} \sum_{i \in \mathcal{J}_N^1(\bX_1, \widehat{\bS}_1)} \Bigl(\bar{\mu}_{1}(\bX_i, \bS_i) - \widehat{\mu}_{1}(\bX_i, \widehat{\bS}_i)\Bigr)\\
       	& \rightarrow \mu_1(\bX_1,\bS_1) - \bar{\mu}_{1}(\bX_1, \bS_1)
       \end{split}
   \end{equation*}
   in probability. Thus
   \begin{equation*}
   	\widehat{\mu}^{c}_{1}(\bX_1,\widehat{\bS}_1) \xrightarrow{p} \mu_1(\bX_1,\bS_1)
   \end{equation*}
   Similarly, we can prove that
   \begin{equation*}
   	\widehat{\mu}^{c}_{0}(\bX_1,\widehat{\bS}_1) \xrightarrow{p} \mu_0(\bX_1,\bS_1)
   \end{equation*}
   This concludes that 
   $$\widehat{\tau}^c(\bX_1, \widehat{\bS}_1) - \tau(\bX_1, \bS_1) \xrightarrow{\mathrm{p}} 0$$
   By setting $w_{j \leftarrow i} = 1$ if $i \in \mathcal{J}_N^{D_i}(\bX_j,\widehat{\bS}_j)$ and 0 otherwise, Assumption \ref{assu:weights} is automatically satisfied and Assumption \ref{assu:weights_density} is satisfied as implied by Theorem B.1 in \cite{lin2023estimation}. Then following Lemma \ref{lem:consistency_bad_generation}, we have
   $$\widehat{\tau}^c - \tau \xrightarrow{\mathrm{p}} 0$$
\end{proof}

\section{Technical Lemmas}
\label{sec:technical_lem}
\begin{lemma}
\label{lem:ci_equivalence}
    Under Assumption \ref{ass:strong_ignorability_DAG} and for any $\mathcal{M}$ such that $\pa(j) \subseteq \mathcal{M} \subseteq j^-$, we have
    \begin{equation*}
        \begin{split}
            (k \rightarrow j) \not \in \mathcal{E} \Rightarrow X_{k} \perp_{\mathbb{P}_{\bX}} X_{j} | \bX_{\mathcal{M}/k}
        \end{split}
    \end{equation*}
\end{lemma}

\begin{proof}[Proof of Lemma \ref{lem:ci_equivalence}]
	This follows from Proposition 1 in \cite{shi2023testing}.
\end{proof}

\begin{lemma}
\label{lem:diffusion_transfer}
    Under Assumption \ref{assu:proxy_DAG}-\ref{assu:target_density}, the error in conditional diffusion generation via transfer learning is
    \begin{align*}
        \mathrm{E}_{\mathcal{D}_t, \mathcal{D}_s} \mathrm{E}_{\bX_{k+},\bS}\left[\mathrm{TV}\left(\mathbb{P}_{X_j \mid \bX_{k+},\bS}, \hat{\mathbb{P}}_{X_j \mid \bX_{k+},\bS}\right)\right]=O\left(n^{-\frac{\beta}{1+d_{h}+2 \beta}} (\log (n))^{\max(19/2,(\beta+2)/2)}+\varepsilon_s\right)
    \end{align*}
    Denote $\boldsymbol{Z} = (\bX_{k+}^\top,\bS^{\top})^\top$, and define the following coefficient
    \begin{align*}
    	\mathcal{T}\left(\boldsymbol{Z}^{\star}\right)=\sup _{\mathbf{\theta} \in \mathcal{F}, h \in \Theta_h} \sqrt{\frac{\int_{\underline{\tau}}^{\bar{\tau}} \mathbb{E}_{X_j \sim \mathbb{P}\left(X_j \mid \boldsymbol{Z}^{\star}\right)} \mathbb{E}_{X_j^{\prime} \sim \mathrm{N}\left(\alpha_t X_j, \sigma_t^2 I\right)}\left[\left\|\mathbf{\theta}\left(X_j^{\prime}, h(\boldsymbol{Z}^{\star}), t\right)-\nabla \log p_t\left(X_j^{\prime} \mid \boldsymbol{Z}^{\star}\right)\right\|^2\right] \mathrm{d} t}{\int_{\underline{\tau}}^{\bar{\tau}} \mathbb{E}_{X_j,\boldsymbol{Z}} \mathbb{E}_{X_j^{\prime} \sim \mathrm{N}\left(\alpha_t X_j, \sigma_t^2 I\right)}\left[\left\|\mathbf{\theta}\left(X_j^{\prime}, h(\boldsymbol{Z}), t\right)-\nabla \log p_t\left(X_j^{\prime} \mid \boldsymbol{Z}\right)\right\|^2\right] \mathrm{d} t}}
    \end{align*}
    then we have
    \begin{align*}
    	\mathrm{E}_{\mathcal{D}_t, \mathcal{D}_s} \left[\mathrm{TV}\left(\mathbb{P}_{X_j \mid \boldsymbol{Z}^\star}, \hat{\mathbb{P}}_{X_j \mid \boldsymbol{Z}^\star}\right)\right]=\mathcal{T}\left(\boldsymbol{Z}^{\star}\right) O\left(n^{-\frac{\beta}{1+d_{h}+2 \beta}} (\log (n))^{\max(19/2,(\beta+2)/2)}+\varepsilon_s\right)
    \end{align*}
\end{lemma}

\begin{proof}[Proof of Lemma \ref{lem:diffusion_transfer}]
	This follows from Theorem 1 in \cite{tian2024enhancing} and Proposition 4.5 in \cite{fu2024unveil}.
\end{proof}

We discuss a generalized version of bias-corrected estimator by
\begin{equation}
	\label{eq:estimator_bc_generalized}
	\widehat{\tau}^{c} = \widehat{\tau} +  \frac{1}{n} [\sum_{i=1,D_i=1}^{n} (\sum_{j = 1}^n \frac{w_{j \leftarrow i}}{w_{j\cdot}^1}\widehat{R}_i) - \sum_{i=1,D_i=0}^{n} (\sum_{j = 1}^n \frac{w_{j \leftarrow i}}{w_{j\cdot}^0}\widehat{R}_i)]
\end{equation}
where $w_{j\cdot}^1 = \sum_{i=1,D_i=1}^{n} w_{j \leftarrow i}$, $w_{j\cdot}^0 = \sum_{i=1,D_i=0}^{n} w_{j \leftarrow i}$ and $w_{i \leftarrow j}$ can be determined in various approaches as long as they satisfy the following assumptions.

\begin{assumption}
	\label{assu:weights}
	Let $\pi': [n] \rightarrow [n]$ be any permutation, where $[n]=\{1,\ldots,n\}$. For samples $\{\bX_i, D_i, Y_i\}_{i=1}^n$ given, let $\{w_{i \leftarrow j}\}_{i,j \in [n]}$ be the weights constructed by $\{\bX_i, D_i, Y_i\}_{i=1}^n$, and $\{w^{\pi'}_{i \leftarrow j}\}_{i,j \in [n]}$ be the weights constructed by $\{\bX_{\pi'(i)}, D_{\pi'(i)}, Y_{\pi'(i)}\}_{i=1}^n$. Then for any $i,j \in [n]$ and any permutation $\pi'$, we have $w_{i \leftarrow j} = w^{\pi'}_{\pi'(i) \leftarrow \pi'(j)}$.
\end{assumption}

\begin{assumption}
	\label{assu:weights_density}
	We denote $e(\bX,\bS) = \mathrm{P}(D=1|\bX,\bS)$ and assume that
	$$\lim _{n \rightarrow \infty} \mathrm{E}\left[\sum_{j=1}^n \frac{w_{j \leftarrow i}}{w_{j \cdot}^{D_i}}-\left(D_{i} \frac{1}{e\left(\bX_i,\bS_i\right)}+\left(1-D_{i}\right) \frac{1}{1-e\left(\bX_i,\bS_i\right)}\right)\right]^{2}=0.$$
\end{assumption}

Assumption \ref{assu:weights} requires the weights to be permutation-invariant and to some extent, Assumption \ref{assu:weights_density} requires that the weights can form a good estimate of density ratio of the treated and controlled samples.

\begin{lemma}
	\label{lem:consistency_bad_generation}
	Under Assumptions \ref{assu:conditional_random}, \ref{assu:proxy}, \ref{assu:data}, \ref{assu:uncorrected_estimator}, \ref{assu:weights}, \ref{assu:weights_density} and $N \rightarrow \infty, \frac{N \log (n)}{n} \rightarrow 0$ as $n, M \rightarrow \infty$, then for the estimate ATE: 
	$$\widehat{\tau}^{\mathrm{c}}-\tau \xrightarrow{\mathrm{p}} 0.$$
\end{lemma}

\begin{proof}[Proof of Lemma \ref{lem:consistency_bad_generation}]
	Let $\bar{R}_i = Y_i - \bar{\mu}_{D_i}(\boldsymbol{X_i},\bS_i)$, then from \eqref{eq:estimator_bc_generalized},
	\begin{equation}
		\label{eq:estimator_bc_decomposition}
		\begin{split}
			\widehat{\tau}^{c} = & \widehat{\tau} +  \frac{1}{n} [\sum_{i=1,D_i=1}^{n} (\sum_{j = 1}^n \frac{w_{j \leftarrow i}}{w_{j\cdot}^1}\widehat{R}_i) - \sum_{i=1,D_i=0}^{n} (\sum_{j = 1}^n \frac{w_{j \leftarrow i}}{w_{j\cdot}^0}\widehat{R}_i)] \\
			= & \frac{1}{n} \sum_{i=1}^n\left[\widehat{\mu}_1\left(\bX_i,\widehat{\bS}_i\right)-\bar{\mu}_1\left(\bX_i,\bS_i\right)\right]-\frac{1}{n} \sum_{i=1}^n\left[\widehat{\mu}_0\left(\bX_i,\widehat{\bS}_i\right)-\bar{\mu}_0\left(\bX_i,\bS_i\right)\right] \\
			& +\frac{1}{n}\left[\sum_{D_i =1}^n \left(\sum_{j = 1}^n \frac{w_{j \leftarrow i}}{w_{j\cdot}^1}\left(\bar{\mu}_{1}-\widehat{\mu}_{1}\right)\right) - \sum_{D_i =0}^n \left(\sum_{j = 1}^n \frac{w_{j \leftarrow i}}{w_{j\cdot}^0}\left(\bar{\mu}_{0}-\widehat{\mu}_{0}\right)\right)\right] \\
			& +\frac{1}{n}\left[\sum_{D_i =1}^n \left(\sum_{j = 1}^n \frac{w_{j \leftarrow i}}{w_{j\cdot}^1}-\frac{1}{e\left(\bX_i,\bS_i\right)}\right) \bar{R}_i -\sum_{D_i =0}^n\left(\sum_{j = 1}^n \frac{w_{j \leftarrow i}}{w_{j\cdot}^0}-\frac{1}{1-e\left(\bX_i,\bS_i\right)}\right) \bar{R}_i\right] \\
			& +\frac{1}{n}\left[\sum_{i=1}^n\left(1-\frac{D_i}{e\left(\bX_i,\bS_i\right)}\right) \bar{\mu}_1\left(\bX_i,\bS_i\right)-\sum_{i=1}^n\left(1-\frac{1-D_i}{1-e\left(\bX_i,\bS_i\right)}\right) \bar{\mu}_0\left(\bX_i,\bS_i\right)\right] \\
			& +\frac{1}{n}\left[\sum_{i=1}^n \frac{D_i}{e\left(\bX_i,\bS_i\right)} Y_i-\sum_{i=1}^n \frac{1-D_i}{1-e\left(\bX_i,\bS_i\right)} Y_i\right],
		\end{split}
	\end{equation}
	For the first term in \eqref{eq:estimator_bc_decomposition}, we have
	\begin{equation*}
		\left | \frac{1}{n} \sum_{i=1}^n\left[\widehat{\mu}_1\left(\bX_i,\widehat{\bS}_i\right)-\bar{\mu}_1\left(\bX_i,\bS_i\right)\right] \right | \leq \max_{i=1,\ldots,n}|\widehat{\mu}_1(\bX_i,\widehat{\bS}_i) - \bar{\mu}_1(\bX_i,\bS_i)| = o_p(1)
	\end{equation*}
	For the second term in \eqref{eq:estimator_bc_decomposition}, we have
	\begin{equation*}
		\begin{split}
			& \left | \frac{1}{n}\sum_{D_i =1}^n \sum_{j = 1}^n \frac{w_{j \leftarrow i}}{w_{j\cdot}^1}\left(\bar{\mu}_{1}\left(\bX_i,\bS_i\right)-\widehat{\mu}_{1}\left(\bX_i,\widehat{\bS}_i\right)\right) \right | \\
			\leq &\max_{i=1,\ldots,n}|\widehat{\mu}_1(\bX_i,\widehat{\bS}_i) - \bar{\mu}_1(\bX_i,\bS_i)|\frac{1}{n}\sum_{D_i =1}^n \left(\sum_{j = 1}^n \frac{w_{j \leftarrow i}}{w_{j\cdot}^1}\right) = \max_{i=1,\ldots,n}|\widehat{\mu}_1(\bX_i,\widehat{\bS}_i) - \bar{\mu}_1(\bX_i,\bS_i)| = o_p(1)
		\end{split}
	\end{equation*}
	since $\frac{1}{n}\sum_{D_i =1}^n \left(\sum_{j = 1}^n \frac{w_{j \leftarrow i}}{w_{j\cdot}^1}\right) = 1$. Similarly, we have
	\begin{equation*}
		\begin{split}
			\left | \frac{1}{n}\sum_{D_i =0}^n \sum_{j = 1}^n \frac{w_{j \leftarrow i}}{w_{j\cdot}^0}\left(\bar{\mu}_{0}\left(\bX_i,\bS_i\right)-\widehat{\mu}_{0}\left(\bX_i,\widehat{\bS}_i\right)\right) \right | = o_p(1)
		\end{split}
	\end{equation*}
	For the third term in \eqref{eq:estimator_bc_decomposition}, by Assumption \ref{assu:weights_density}, we have
	\begin{equation*}
		\begin{split}
			& \E \left | \frac{1}{n} \sum_{D_i =1}^n \left(\sum_{j = 1}^n \frac{w_{j \leftarrow i}}{w_{j\cdot}^1}-\frac{1}{e\left(\bX_i,\bS_i\right)}\right) \bar{R}_i  \right | \\
			\leq & \sqrt{\E \left[\sum_{j = 1}^n \frac{w_{j \leftarrow i}}{w_{j\cdot}^1}-\frac{1}{e\left(\bX_i,\bS_i\right)}\right]^2} \sqrt{\E\left[D_i \bar{R}_i\right]^2} \\
			\leq & \sqrt{\E \left[\sum_{j = 1}^n \frac{w_{j \leftarrow i}}{w_{j\cdot}^1}-\frac{1}{e\left(\bX_i,\bS_i\right)}\right]^2} \sqrt{\E\left[U_1^2\right] + \E\left[\widehat{\mu}_1\left(\bX_i,\widehat{\bS}_i\right)-\bar{\mu}_1\left(\bX_i,\bS_i\right)\right]^2} = o(1)
		\end{split}
	\end{equation*}
	Similarly, we have 
	\begin{equation*}
		\begin{split}
			\E \left | \frac{1}{n} \sum_{D_i =0}^n \left(\sum_{j = 1}^n \frac{w_{j \leftarrow i}}{w_{j\cdot}^0}-\frac{1}{1- e\left(\bX_i,\bS_i\right)}\right) \bar{R}_i  \right | = o(1)
		\end{split}
	\end{equation*}
	For the fourth term in \eqref{eq:estimator_bc_decomposition}, notice that
	\begin{equation*}
		\E\left[\frac{1}{n}\sum_{i=1}^n\left(1-\frac{D_i}{e\left(\bX_i,\bS_i\right)}\right) \bar{\mu}_1\left(\bX_i,\bS_i\right)\mid \bX_1,\ldots,\bX_n ,\bS_1,\ldots,\bS_n\right] = 0
	\end{equation*}
	and 
	\begin{equation*}
		\begin{split}
			& \operatorname{Var}\left[\frac{1}{n}\sum_{i=1}^n\left(1-\frac{D_i}{e\left(\bX_i,\bS_i\right)}\right) \bar{\mu}_1\left(\bX_i,\bS_i\right)\right] \\
			= & \E \left[\operatorname{Var}\left[\frac{1}{n}\sum_{i=1}^n\left(1-\frac{D_i}{e\left(\bX_i,\bS_i\right)}\right) \bar{\mu}_1\left(\bX_i,\bS_i\right)\mid \bX_1,\ldots,\bX_n,\bS_1,\ldots,\bS_n \right]\right] \\
			= & \frac{1}{n} \E\left[ \bar{\mu}^2_1\left(\bX_i,\bS_i\right)(\frac{1}{e\left(\bX_i,\bS_i\right)}-1)\right] = o(1)
		\end{split}
	\end{equation*}
	then 
	\begin{equation*}
		\frac{1}{n}\sum_{i=1}^n\left(1-\frac{D_i}{e\left(\bX_i,\bS_i\right)}\right) \bar{\mu}_1\left(\bX_i,\bS_i\right) = o_p(1)
	\end{equation*}
	Similarly, we have
	\begin{equation*}
		\frac{1}{n} \sum_{i=1}^n\left(1-\frac{1-D_i}{1-e\left(\bX_i,\bS_i\right)}\right) \bar{\mu}_0\left(\bX_i,\bS_i\right) = o_p(1)
	\end{equation*}
	For the fifth term in \eqref{eq:estimator_bc_decomposition}, we have
	\begin{equation*}
		\frac{1}{n}\left[\sum_{i=1}^n \frac{D_i}{e\left(\bX_i,\bS_i\right)} Y_i-\sum_{i=1}^n \frac{1-D_i}{1-e\left(\bX_i,\bS_i\right)} Y_i\right] \xrightarrow{p} \tau
	\end{equation*}
	by the weak law of large numbers. This completes the proof.
\end{proof}

\bibliography{reference}

\end{document}